\documentclass{article}%[letter,11pt,final]
\usepackage{amsmath,amssymb,amsthm}
\usepackage{bm}
\usepackage{bbm}
\usepackage{bbold}
\usepackage{booktabs}
\usepackage{caption}
\usepackage{enumerate}
\usepackage{float}
\usepackage[colorlinks,citecolor=blue,linkcolor=BrickRed]{hyperref}
\usepackage[nameinlink]{cleveref}
\usepackage[utf8]{inputenc}
\usepackage{subcaption}
\usepackage[normalem]{ulem}
\usepackage[usenames,dvipsnames]{xcolor}
\usepackage{xspace}
\usepackage[numbers,sort&compress]{natbib}
\usepackage{lipsum}
\usepackage{thmtools, thm-restate}
\usepackage{algorithm}
\usepackage{algcompatible}

\newtheorem{theorem}{Theorem}

\newtheorem{lemma}[theorem]{Lemma}
\newtheorem{fact}[theorem]{Fact}

\newtheorem{definition}[theorem]{Definition}
\newtheorem{corollary}[theorem]{Corollary}
\newtheorem{example}[theorem]{Example}

\newcounter{note}[section]

\newcommand{\R}{\mathbb{R}}%Expectation

\newcommand{\poly}{\text{poly}}
\newcommand{\calP}{\mathcal{P}}

\allowdisplaybreaks

\usepackage{fullpage}
\usepackage{graphicx}

\newcommand{\eps}{\varepsilon}
\newcommand{\mr}[1]{\mathrm{#1}}

\newcommand{\tl}[1]{\widetilde{#1}}

\newcommand{\tO}{\tl{O}}
\newcommand{\1}[1]{\mathbb{1}[#1]}

\newcommand{\congest}{$\mathsf{CONGEST}\,$}

\newcommand{\virt}{\ensuremath{_{\mathrm{virt}}}\xspace}
\newcommand{\orig}{\ensuremath{_{\mathrm{orig}}}\xspace}

\newcommand{\real}{\ensuremath{_{\mathrm{real}}}\xspace}

\newcommand{\up}{\ensuremath{_{\text{up}}}\xspace}
\newcommand{\down}{\ensuremath{_{\text{down}}}\xspace}
\newcommand{\nquad}{\!\!\!}

\newcommand{\Cov}{\mathrm{Cov}}

\newcommand{\CrossCov}{\mathrm{CrossCov}}
\newcommand{\Cut}{\mathrm{Cut}}

\newcommand{\parent}{\mathrm{parent}}
\newcommand{\depth}{\mathrm{depth}}
\newcommand{\LCA}{\mathrm{LCA}}

\newcommand{\HLdepth}{\text{HL-depth}}
\newcommand{\subtree}{\mathrm{subtree}}
\newcommand{\mytop}{\mathrm{top}} %% CHECK: DO NOT USE \top
\newcommand{\bottom}{\mathrm{bottom}}

\newcommand{\anc}{\mathrm{anc}}
\newcommand{\desc}{\mathrm{desc}}
\newcommand{\SQ}{\mathrm{SQ}}

\title{Universally-Optimal Distributed Exact Min-Cut\footnote{Supported in part by funding from the European Research Council (ERC) under the European Union's Horizon 2020 research and innovation program (grant agreement No. 853109) and the Swiss National Foundation (project grant 200021-184735).}}

\author{Mohsen Ghaffari \and Goran Zuzic}

\date{\today}

\begin{document}

\maketitle

\begin{abstract}
  We present a universally-optimal distributed algorithm for the exact weighted min-cut. The algorithm is guaranteed to complete in $\widetilde{O}(D + \sqrt{n})$ rounds on every graph, recovering the recent result of Dory, Efron, Mukhopadhyay, and Nanongkai~[STOC'21], but runs much faster on structured graphs. Specifically, the algorithm completes in $\widetilde{O}(D)$ rounds on (weighted) planar graphs or, more generally, any (weighted) excluded-minor family.

  \smallskip

  We obtain this result by designing an aggregation-based algorithm: each node receives only an aggregate of the messages sent to it. While somewhat restrictive, recent work shows any such black-box algorithm can be simulated on any minor of the communication network. Furthermore, we observe this also allows for the addition of (a small number of) arbitrarily-connected virtual nodes to the network. We leverage these capabilities to design a min-cut algorithm that is significantly simpler compared to prior distributed work. We hope this paper showcases how working within this paradigm yields simple-to-design and ultra-efficient distributed algorithms for global problems.

  \smallskip

  Our main technical contribution is a distributed algorithm that, given any tree $T$, computes the minimum cut that $2$-respects $T$ (i.e., cuts at most $2$ edges of $T$) in universally near-optimal time. Moreover, our algorithm gives a \emph{deterministic} $\widetilde{O}(D)$-round 2-respecting cut solution for excluded-minor families and a \emph{deterministic} $\widetilde{O}(D + \sqrt{n})$-round solution for general graphs, the latter resolving a question of Dory, et al.~[STOC'21]
\end{abstract}
\thispagestyle{empty}

\newpage
\tableofcontents
\bigskip
\thispagestyle{empty} % in case you don't want to number this page

\newpage

\section{Introduction}\setcounter{page}{1}\label{sec:intro}
Computing the \emph{minimum cut} in a graph is one of the fundamental and well-studied graph problems. This problem asks for computing the smallest collection of edges, in terms of their number in unweighted graphs and in terms of the total sum of their weights in weighted graphs, whose removal would disconnect the graph. This notion captures important properties of the network such as its \emph{robustness to failure}---e.g., how many link failures can the network withstand before it gets disconnected--- or \emph{communication bottlenecks}---e.g., the smallest capacity of links connecting one set of nodes to the rest of the network. Over the past decade, we have witnessed significant developments on this problem in the distributed computing setting. To review these results, let us first recall the message-passing model of distributed graph algorithms.

\paragraph{Model.} As standard, we work with the standard message-passing model of distributed computing, often referred to as the \congest model~\cite{peleg2000distributed}. The network is abstracted as an $n$-node connected undirected graph $G=(V, E)$ where each node represents one of the computers in the network (i.e., has its own processor and private memory). Communication takes place in synchronous rounds and per round, each node can send one $O(\log n)$-bit message to each of its neighbors. The nodes do not know the topology of the network at the start of the algorithm (except for each knowing its own neighbors, and perhaps some estimates on the total number of nodes $n$ and the network diameter $D$). Initially, nodes only know their unique $O(\log n)$-bit ID and the IDs of adjacent nodes. In the end, each node $v$ should know its own part of the output, e.g., the size of the minimum cut and which adjacent edges are in the computed cut.

\paragraph{State of the art on distributed computation of min-cut.} The initial progress on distributed algorithms on min-cut focused on approximations. Ghaffari and Kuhn~\cite{ghaffari2013distributed} gave a randomized algorithm that computes a $2+\eps$ approximation, for an arbitrarily small positive constant $\eps$, of minimum cut in $\tl{O}(D+\sqrt{n})$ rounds for weighted graphs. They also showed, by a minor adaptation of the lower bound of Das Sarma et al.~\cite{dassarma2012distributed}, that any non-trivial approximation of minimum cut in weighted graphs requires $\tl{\Omega}(D+\sqrt{n})$ rounds. For unweighted graphs, their lower bound degrades to $\tl{\Omega}(D+\sqrt{n/\lambda})$ where $\lambda$ denotes the minimum cut size. Nanongkai and Su~\cite{nanongkai2014almost} improved the approximation factor to a $1+\eps$ while maintaining the same $\tl{O}(D+\sqrt{n})$ round complexity. Progress on exact computation was more scarce, until a result of Daga, Henzinger, Nanongkai, and Saranurak~\cite{daga2019distributed} that obtained the first sublinear-time algorithm for unweighted graphs. Concretely, their algorithm computes the exact minimum cut in $\tl{O}(n^{1-1/353}D^{1/353} + n^{1-1/706})$ rounds in unweighted graphs. Ghaffari, Nowicki, and Thorup~\cite{ghaffari2020faster} then provided a different exact algorithm for unweighted graphs that improved the round complexity further to $\tl{O}(n^{0.8}D^{0.2} + n^{0.9})$. Also, Parter~\cite{parter2019small} gave an algorithm with round complexity $(\lambda D)^{O(\lambda)}$ for computing the exact unweighted min-cut, where $\lambda$ denotes the min-cut size; this, in particular, runs in $\poly(D)$ for unweighted graphs with constant min-cut size. Finally, in a recent breakthrough, Dory, Efron, Mukhopadhyay, and Nanongkai~\cite{dory2021distributed} presented an algorithm that achieves the worst-case optimal round complexity of $\tl{O}(D+\sqrt{n})$ for exact computation of minimum cut in unweighted graphs.

\paragraph{Beyond worst-case.} When can we call a distributed algorithm ``optimal'' or ``near-optimal'' and what exactly do we mean by that? The $\tl{O}(D+\sqrt{n})$ complexity achieved above is near-optimal, in a worst-case sense, as follows: there is a weighted graph with diameter $D=O(\log n)$ in which any min-cut algorithm would need $\tl{\Omega}(D+\sqrt{n})$ rounds. This optimality is stronger than another worst-case optimality, where we would consider $\tl{O}(n)$-round algorithms \emph{near-optimal}. Notice that the latter is also a correct statement, as there is a graph in which any algorithm needs $\Omega(n)$ rounds (namely, a simple $n$-node cycle). However, the former gives a sharper bound for a wide range of graphs of interest, particularly, graphs where the diameter $D$ is small. Is there an even stronger notion of optimality?

One could think about focusing on particular graph parameters that capture ``usual'' network graphs and aim for faster algorithms when these parameters are small. Even then, we are essentially justifying the performance of the algorithm on any network $G$ of the family, because of the mere existence of one (concocted) network $G'$ in the family where the algorithm cannot perform faster. Plausibly, in most usages of the algorithm, the network $G$ is much more well-behaved than that tailored worst-case graph $G'$, and thus we could desire much faster algorithms. 

\paragraph{Universal Optimality.} A far more ambitious goal is to seek \emph{universal optimality}. That is, to seek a single (i.e., uniform) algorithm which, when run on a network $G$, has the time-complexity that is competitive with the fastest (correct) algorithm's round complexity on that particular network $G$ itself. This paper's objective is to develop such a universally near-optimal algorithm for exact computation of minimum cut. Toward this goal, let us briefly take a detour and recall the concept of low-congestion shortcuts.

\paragraph{Detour to low-congestion shortcuts and shortcut quality.} Given a network $G=(V, E)$, Ghaffari and Haeupler~\cite{ghaffari2016distributed} defined the \emph{shortcut quality} $\SQ(G)$ as the smallest value $Q$ such that we have the following: for any (adversarial) partition of vertices $V$ into disjoint \emph{parts} $V_1, V_2, \ldots, V_N$, each of which induces a connected subgraph $G[V_i]$, there exists a collection of subgraphs $H_1, H_2, \ldots, H_N$,  such that (1) for each $i\in [1, N]$ the diameter of $G[V_i] \cup H_i$ is at most $Q$, and (2) each edge $e\in E$ appears in at most $Q$ many of the subgraphs $H_i$. The graph $H_i$ is called the \emph{shortcut} for part $V_i$.

Ghaffari and Haeupler~\cite{ghaffari2016distributed} showed that any $D$-diameter $n$-node graph admits a shortcut with quality $O(D+\sqrt{n})$,
and they showed algorithms with round complexity $\tl{O}(\SQ(G))$ for exact computation of minimum spanning tree and $1+\eps$ approximation of minimum cut in weighted graphs. These assume that shortcuts can be computed in $\tl{O}(\SQ(G))$, and otherwise, the construction time should be added to the complexity. This result immediately recovers the $\tl{O}(D+\sqrt{n})$ complexity of the minimum spanning tree and $1+\eps$ approximation of minimum cut in general weighted graphs with hop-diameter $D$ and $n$ nodes. But it also leads to significantly faster algorithms in more well-behaved graphs.

In particular, Ghaffari and Haeupler~\cite{ghaffari2016distributed} showed that the shortcut quality $\SQ(G)$ is smaller for many other graph families. For instance, $\SQ(G) = \tl{O}(D)$ for any planar graph or constant-genus family. This was later sharpened and extended for graphs with bounded genus, bounded treewidth, and bounded pathwidth~\cite{haeupler2016near}. Haeupler, Li, and Zuzic~\cite{haeupler2018minor} gave shortcuts for excluded-minor graphs, with quality and construction-time $\tl{O}(D^2)$. Finally, Ghaffari and Haeupler~\cite{ghaffari2021excluded} improved and strengthened all these results and showed that excluded-minor graphs, which contain all previously mentioned graphs, admit shortcuts with quality $\tl{O}(D)$. For all of the aforementioned results, it is known how to construct shortcuts of quality $\tl{O}(\SQ(G))$ with an efficient $\tl{O}(\SQ(G))$-round and \emph{deterministic} distributed algorithm~\cite{haeupler2016low,haeupler2018round,ghaffari2021excluded}. Hence, these imply an $\tl{O}(D)$-round algorithm for $(1+\eps)$-approximation of weighted min-cut in any $D$-diameter excluded-minor graph network.

These results focus on mostly on sparse graphs, in a vague sense. On the opposite side, for well-connected graphs, the results of Ghaffari, Kuhn, and Su~\cite{ghaffari2017MixingTime}, which were sharpened by Ghaffari and Li~\cite{ghaffari2018new} showed that any graph with $1/\poly(\log n)$ mixing time admits a shortcut with quality $\poly(\log n)$ and one can compute shortcut of quality $2^{O(\sqrt{\log n})}$ in them in $2^{O(\sqrt{\log n})}$ rounds (indeed in any graph with mixing time $2^{- O(\sqrt{\log n})}$). These imply an $2^{O(\sqrt{\log n})}$-round algorithm for $(1+\eps)$-approximation of weighted min-cut in well-connected graphs, with mixing-time $1/\poly(\log n)$ or even $2^{- O(\sqrt{\log n})}$. This in particular includes Erdos-Renyi random graphs above the connectivity threshold. 

\paragraph{Back to Universal Optimality.} Haeupler, Wajc, and Zuzic~\cite{haeupler2020shortcuts} showed that the shortcut quality is not only an upper bound for the round complexity of computing a minimum spanning tree or approximation of min-cut, as shown in \cite{ghaffari2016distributed} but also a \emph{universal} lower bound for it. That is, roughly speaking, for any network graph $G$ with shortcut quality $\SQ(G)$, one can show that any distributed algorithm that works correctly on all graphs needs $\tl{\Omega}(\SQ(G))$ rounds to solve the (both approximate or exact) minimum cut problem on the network $G$ itself. Notice that in this statement, while the topology network $G$ itself is fixed, and can be even known to all the nodes, the weights on the edges of $G$ are the input to the problem.
%Hence, in particular, an algorithm that computes the minimum cut in $\tl{O}(\SQ(G))$ rounds is a universally near-optimal algorithm, as any .

In light of this, we can say that the $(1+\eps)$-min-cut approximation algorithm of Ghaffari and Haeupler~\cite{ghaffari2016distributed}, which runs in $\tl{O}(\SQ(G))$ rounds once given the shortcuts, is a universally-optimal algorithm for $(1+\eps)$-approximation min-cut since any correct algorithm requires $\tl{\Omega}(\SQ(G))$ rounds. This is modulo one small but important issue: the time for computing shortcuts has not been taken into account. However, arguably, that is an orthogonal topic within the ambitious path toward the holy grail of obtaining universally-optimal distributed algorithms for all global graph problems: we can separate the issue of efficient shortcut computation in different graphs from the issue of how to design algorithms for various problems whose complexity is proportional to the shortcut quality once efficient computation is assumed.
 
Only the latter part is within the scope of this paper. A universally-optimal min-cut algorithm that assumes efficient construction still implies an unconditional universally-optimal algorithm when the network is guaranteed to not contain a fixed minor, and for the setting with known topology (where the weights are still unknown and a part of the input, also known as the supported CONGEST) as studied by Haeupler, Wajc, and Zuzic~\cite{haeupler2020shortcuts}.
% . In particular, the aforementioned algorithm of \cite{ghaffari2016distributed} is a universally-optimal algorithm for $(1+\eps)$-approximate min-cut in the setting with known topologies.
Furthermore, for the more standard unknown-topology setting, there has been significant recent progress on the former issue of fast construction of shortcuts: in particular, Haeupler \textcircled{r} Raecke \textcircled{r} Ghaffari~\cite{haeupler2022oblRouting} showed that one can obtain shortcuts with quality $\poly(\SQ(G)) n^{o(1)}$ in the same number of rounds. As such, combining this with \cite{ghaffari2016distributed}, we can now compute a $(1+\eps)$-approximation of min-cut in $\poly(\SQ(G)) n^{o(1)}$ in any network $G$, and this is within a polynomial of the best-possible bound for the network $G$ itself, modulo an $n^{o(1)}$ factor. Moreover, any future $\tl{O}(\SQ(G))$-quality construction in $\tl{O}(\SQ(G))$ rounds would retroactively turn these conditional universally-optimal algorithms into unconditional ones.

\paragraph{The Minor-Aggregation model.} State-of-the-art distributed algorithms have become increasingly more complex, due to the influx of new ideas and their increasing complexity. To address this issue, Zuzic~\textcircled{r}~al.~\cite{goranci2022universally} introduced the \emph{Minor-Aggregation model}: a simple and powerful interface for designing ultra-fast distributed algorithms in the standard message-passing (i.e., CONGEST) model. The interface provides high-level primitives that simplify algorithm design; the primitives are then, in turn, efficiently implemented in CONGEST using low-congestion shortcuts and the multitude of tools developed around them. For example, \cite{goranci2022universally} used the interface to simplify the design of their universally-optimal $(1+\eps)$-approximate distributed shortest path algorithm. On a technical level, the Minor-Aggregation model restricts the algorithm to operate only on aggregates: each node receives only an aggregate value (e.g., sum, max, logical-OR, etc.) of all messages sent to it. However, this restriction, combined with low-congestion shortcuts, enables efficient edge contractions which are difficult to efficiently implement in a distributed setting. In other words, this allows a black-box algorithm to be run on an arbitrary minor. As an instructive example, consider the classic Boruvka's MST algorithm which works by computing the minimum-weight outgoing edge from each node and then contracting all such edges; this iteration is repeated for $O(\log n)$ steps until the graph is trivial. Boruvka's algorithm naturally operates on aggregates, hence it can be immediately performed on the minor resulting from contracting minimum-weighted edges, giving us an $O(\log n)$-round Minor-Aggregation algorithm. Using prior work, this can be turned into, say, an $\tl{O}(D)$-round algorithm for (weighted) planar networks. More generally, a $\tau_1$-round Minor-Aggregation algorithm can be turned into an $\tl{O}(\tau_1 \cdot \tau_2)$-round CONGEST algorithm, where $\tau_2$ is the time to construct shortcuts of quality $\tau_2$ (see below for a list of implications).

% This is achieved by black-boxing the recent advancements in low-congestion shortcuts and universal optimality in a user-friendly way. On a technical level, the interface is built upon the observation that many important distributed algorithms only compute aggregates (e.g., sums, max, ORs, etc.) of their neighbors, and on the observation that such algorithms can be executed on an arbitrary minor of the communication network with the help of low-congestion shortcuts.

\paragraph{Our contributions.}
Our first contribution of this paper is to develop a $\poly(\log n)$-round Minor-Aggregation algorithm that computes the exact minimum cut in weighted graphs.
% an algorithm that computes the exact minimum cut in weighted graphs, in $\tl{O}(\SQ(G))$ rounds, assuming shortcuts can be efficiently approximated.
This \emph{unconditionally} recovers the breakthrough $\tl{O}(D + \sqrt{n})$-round CONGEST algorithm for general graphs of Dory et al.~\cite{dory2021distributed}. Moreover, it \emph{unconditionally} implies the following set of novel results.%these novel results: a universally-optimal $\tl{O}(D)$-round algorithm for weighted excluded-minor graphs (e.g., weighted planar networks), a universally-optimal $\tl{O}(\SQ(G))$-round algorithm when the topology is known, an almost-universally-optimal $2^{O(\sqrt{\log n})}$ algorithm in well-connected graph with mixing-time $2^{O(\sqrt{\log n})}$, and an almost-universally-optimal $n^{o(1)}$-round algorithms when $\SQ(G) \le n^{o(1)}$.

\begin{theorem}\label{theorem:final-min-cut-result}
  Suppose $G$ is an $n$-node graph with hop-diameter $D$. There are randomized distributed CONGEST algorithms $A_1, A_2, A_3, A_4$ over $G$ for the exact weighted min-cut problem with the following guarantees:
  \begin{itemize}\setlength\itemsep{0em}
  \item When $G$ is an excluded-minor graph (e.g., a planar network), $A_1$ terminates in universally-optimal $\tl{O}(D)$ rounds.
  \item When the graph topology $G$ is known to all nodes, $A_2$ terminates in universally-optimal $\tl{O}(\SQ(G))$ rounds. (Note: this bullet, along with known implications, implies all other bullets.)
  \item When $G$ well-connected graph with mixing-time $2^{O(\sqrt{\log n})}$, $A_3$ terminates in almost-universally-optimal $2^{O(\sqrt{\log n})}$ rounds.
  \item When $\SQ(G) \le n^{o(1)}$, $A_4$ terminates in (almost-universally-optimal) $n^{o(1)}$ rounds.
  \end{itemize}
\end{theorem}
%
% with the universally-near-optimal round complexity $\tl{O}(\SQ(G))$ for the setting with known topology. Moreover, for the setting with unknown topologies, it readily recovers the $\tl{O}(D+\sqrt{n})$ round complexity of the recent Dory et al.~\cite{dory2021distributed} breakthrough for general graphs, but also gives much faster algorithms in more well-behaved networks. In particular, it has round complexity $\tl{O}(D)$ in all excluded-minor weighted graphs, including any planar or constant genus graph, and round complexity $2^{O(\sqrt{\log n})}$ algorithm in well-connected graph with mixing-time $1/\poly(\log n)$ or even $2^{- O(\sqrt{\log n})}$. Finally, coupled with the general shortcut computation algorithm of Haeupler et al.~\cite{haeupler2022oblRouting}, it implies a $\poly(\SQ(G)) \cdot n^{o(1)}$-round exact algorithm for minimum cut on any weighted network $G$, which is within a polynomial of the best-possible bound for the network $G$ itself, modulo an $n^{o(1)}$ factor. In particular, this implies a $n^{o(1)}$-round algorithm in any network where such an algorithm is possible.

%and derandomization of

Our second contribution is a simple but powerful extension of the \emph{Minor-Aggregation model}.
% , which provides a simple and powerful interface for designing ultra-fast distributed algorithms in the standard message-passing (i.e., CONGEST) model.
Specifically, we show that any black-box Minor-Aggregation algorithm can be logically executed on a network graph $G$ \emph{adjoined with (a small number of) arbitrarily-connected virtual nodes which do not need to exist in $G$}; this is compiled down to an algorithm that only communicates using the existing links in the underlying network graph $G$ while suffering only a small overhead. Any such property fails for general CONGEST (i.e., non-aggregation based) algorithms, as adding a single fully-connected virtual node greatly increases the computational power of the model.
% On the other hand, aggregation-based algorithms have the benefit of adding virtual nodes while suffering only a small overhead, allowing us to
Virtual nodes allow us to import various techniques like divide-and-conquer from the centralized and parallel settings into the distributed world, which is the reason why this paper can simplify and speed up the arguably-complicated exact min-cut algorithm of \cite{dory2021distributed}. Moreover, our virtual-node extension has already found prolific use in  Rozhon~\textcircled{r}~al.~\cite{rozhon2022undirected}, which gives a unified algorithm for the deterministic $(1+\eps)$-shortest path that is both the first near-optimal in the parallel setting and the first universally optimal in the distributed setting.
% . We can now distributedly recurse on different parts of the input graph even when the parts are non-disjoint (but are near-disjoint), and even if one needs to slightly change the underlying graph between each call. We heavily exploit such design techniques in this paper allowing us to significantly simplify and speed-up the exposition as compared to Dory et al.

Our third contribution is a \emph{deterministic} Minor-Aggregation algorithm for the 2-respecting min-cut problem, which often implies a deterministic CONGEST algorithm. To give context, our exact min-cut algorithm follows the strategy outlined by Karger~\cite{karger2000minimum}---and frequently used later, e.g,~\cite{daga2019distributed, dory2021distributed}---which is comprised of exactly two self-contained pieces: the \emph{tree packing} and the \emph{2-respecting min-cut}. The former piece, tree packing, is about finding a collection of $\poly(\log n)$ spannings such that every min-cut $2$-respects one tree $T$ in the collection, in the sense that the cut includes at most $2$ edges of $T$. The latter piece, $2$-respecting min-cut, is when we are given a tree $T$ and we should compute the minimum cut in graph $G$ among those that $2$-respect $T$. We note that obtaining ``efficient'' deterministic tree packing is still an active area of research even in the centralized setting (with some exciting recent progress by Li~\cite{li2021deterministic}). In contrast, computing 2-respecting min-cut has been successfully derandomized in the centralized~\cite{gawrychowski2021note} and parallel settings~\cite{lopez2021work}. We contribute the analogous distributed result and obtain a deterministic CONGEST 2-respecting min-cut that terminates in $\tl{O}(D)$ rounds for weighted excluded-minors graphs (e.g., weighted planar graphs), and $\tl{O}(D + \sqrt{n})$ rounds for general graphs. The latter result resolves an open question of \cite{dory2021distributed}, who asked for a $\tl{O}(D + \sqrt{n})$-round deterministic algorithm for this $2$-respecting min-cut problem. To achieve this result, we contribute to the low-congestion shortcut ecosystem of tools by derandomizing several important primitives like heavy-light decompositions, subtree sums, and ancestor sums of trees.

%Furthermore, within the Minor-Aggregation model, we derandomize important primitives like heavy-light decompositions and subtree sums as well as give a deterministic Minor-Aggregation for the 2-respecting tree cuts.\mtodo{none of these has been defined, and especially $2$-respecting tree cut needs proper explanation} However, this only implies the existence of deterministic CONGEST algorithms when shortcuts can be efficiently approximated in a deterministic manner: which is true for general graphs if one aims for $\tl{O}(D + \sqrt{n})$ round complexity, or universally-optimal algorithms in graphs excluding dense minors (e.g., planar, treewidth-bounded, excluded-minor graphs). To this end, our unconditional $\tl{O}(D + \sqrt{n})$ 2-respecting tree cut algorithm for general graphs is deterministic, resolving the open question of Dory et al.~\cite{dory2021distributed}. Moreover, both the virtual-node extension and derandomization of the model has already found prolific use in \cite{rozhon2022undirected} which gives a unified algorithm for the deterministic $(1+\eps)$-shortest path that is both the first near-optimal in the parallel setting and the first universally-optimal in the distributed setting.
% \dots \textcolor{blue}{Write about the model and how it not only simplifies the min-cut algorithm but also provides a convenient framework for other algorithm design problems in \congest, including forward pointers to the follow-up work}

\paragraph{Other related work on exact min-cut.} Algorithms for min-cut have seen a flurry of recent progress. Mukhopadhyay and Nanongkai~\cite{mukhopadhyay2020weighted} observed several structural properties of min-cut that enable the min-cut to be computed more efficiently and in different models. Specifically, they obtain a sequential $O(m \log^2 n + n\log^6 n)$-time algorithm, which compares favorably to the celebrated sequential $O(m\log^3 n)$-time algorithm of Karger~\cite{karger2000minimum}. Subsequent result include work-optimal parallel algorithms for non-sparse graphs~\cite{lopez2021work}, near-existentially-optimal distributed algorithms~\cite{dory2021distributed}, faster directed algorithm for directed min-cut~\cite{cen2021minimum}, etc. At the same time and independently of \cite{mukhopadhyay2020weighted}, Gawrychowski, Mozes, and Weimann~\cite{gawrychowski_et_al:LIPIcs:2020:12464} proposed an $O(m \log^2 n)$-time algorithm for the 2-respecting min-cut in the centralized setting that is deterministic and faster than that of Karger~\cite{karger2000minimum}, and they strengthened and simplified the approach of Mukhopadhyay and Nanongkai to obtain a sequential $O(m \log^2 n + n\log^3 n)$-time algorithm~\cite{gawrychowski2021note}.

\section{An Overview of Our Methods}
% \alert{might be a good idea to give forward section references here}

\noindent\paragraph{Minor-Aggregation with virtual nodes.} We start by giving a short and informal preliminary on the Minor-Aggregation model, as introduced in \cite{goranci2022universally} (see \Cref{sec:prelim-minor-aggregation} for a formal discussion). A distributed algorithm in this model performs computations in synchronous rounds. In each round, the algorithm first contracts an arbitrary subset of edges. Then, each super-node (which is created from contracting a connected component of the contracted edges) sends a message to all of its neighbors. On the receiving end, each node $v$, instead of receiving each of the messages $m_1, \ldots, m_k$ sent to it, receives only an aggregate value $\bigoplus_{i=1}^k m_i$, where $\bigoplus$ is some \emph{aggregate function} like the sum or the max, but can also be as complicated as an arbitrary mergeable sketch. Several observations are immediate: we can run black-box algorithms on minors (due to contractions), and we can run simultaneous algorithms on node-disjoint connected subgraphs (we add the warning that edge-disjointedness would not suffice). The goal is to find a $\poly(\log n)$-round min-cut algorithm in this model, which corresponds to a universally-optimal algorithm (under certain conditions orthogonal to this paper). We contribute to the model in the following ways:
\begin{itemize}
\item \emph{Virtual nodes.} (\Cref{sec:virtual-nodes}) We observe the simple but powerful property that aggregation-based algorithms behave remarkably well under the addition of \emph{virtual nodes}. Specifically, we allow to add $\poly(\log n)$ virtual nodes to the underlying network and arbitrarily connect them with virtual edges, either among themselves or between virtual nodes and nodes of $G$. Any algorithm on the resulting \emph{virtual graph} can be simulated on $G$ with a $\poly(\log n)$ multiplicative blowup in the number of rounds. Note that no such property exists for CONGEST without introducing polynomial blowup factors in the computation. 

\item \emph{Modeling power and caveats when using virtual nodes.} This possibility of adding virtual nodes greatly enhances the modeling power of the Minor-Aggregation model. For example, one can turn a black-box single-source shortest path algorithm in the Minor-Aggregation model into a multi-source shortest path algorithm by creating a virtual super-source and connecting it to a set of source nodes. Moreover, when combined with recursions, virtual nodes can be used to bring many recursive graph algorithms from the centralized and parallel settings into the distributed world. To see why, when doing a recursive call, one often needs to change the subgraphs before passing them to the recursive calls. Virtual nodes provide a very simple way of achieving this. However, we should add a warning about an issue we call \emph{simulation cascade} that can arise when combining virtual nodes and recursions: when issuing a recursive call on a virtual graph, that call has to eventually run on the underlying communication network. A naive solution would be to simply remove the virtual nodes from the recursive call via simulation, thereby causing a (small) multiplicative blowup. However, this multiplicative blowup happens on every level of the recursion, preventing the final algorithm from having a polylogarithmic running time. In this paper, we develop several different solutions for this issue (explained later).

  % Specifically, let $G$ be a undirected graph (representing the network) and $v \in V(G)$ be an arbitrary node. Denote with $G - v$ the graph $G$ with $v$ and its incident edges removed. Now, any Minor-Aggregation algorithm on a graph $G$ can be simulated with a $O(1)$-round Minor-Aggregation algorithm on $G - v$, assuming $G - v$ is connected. 

\item \emph{Deterministic primitives.} (\Cref{sec:det-ops-and-simulation} and \Cref{sec:tree-primitives}) Important primitives like \emph{heavy-light decompositions}, \emph{ancestor sums}, and \emph{subtree sums} of trees are often ubiquitously-used primitives within the low-congestion shortcut framework. Within the Minor-Aggregation model, consider combining subtree sums with the approximate heavy-hitter sketch (which is a mergeable sketch, hence is a valid aggregation operator): given inputs $x_u$ for each node $u$, each node $v$ can compute the heavy hitters among $\{ x_u : u \text{ is in the subtree of } v\}$ (\Cref{example:heavy-hitters}). However, prior work has typically resorted to a randomized implementation of these primitives~\cite{ghaffari2016mst,dory2019improved,haeupler2020shortcuts}. We address this issue by providing deterministic $\poly(\log n)$-round Minor-Aggregation algorithms for all the aforementioned primitives. This yields fully deterministic $\tl{O}(D)$-round CONGEST algorithms for these primitives in excluded-minor networks, and $\tl{O}(D+\sqrt{n})$-round CONGEST algorithms for general graphs. We achieve this result by replacing the randomized star-merging technique used throughout the low-congestion shortcut framework with a deterministic version by leveraging the deterministic 3-coloring of out-degree-one graphs developed by Cole and Vishkin~\cite{cole1986deterministic}.

\end{itemize}

\noindent\paragraph{Minimum cut via tree packing, and $2$-respecting min-cuts.}\nquad(\Cref{sec:tree-packing}) To solve the minimum cut problem, thanks to the known tree packing results~\cite{karger2000minimum, daga2019distributed, dory2021distributed}, it suffices to find the minimum cut among the cuts that $2$-respects a given tree. We note that the algorithm from prior work for this tree-packing part easily extends to our setting, as we explain in \Cref{thm:treePacking}. Our focus will be on computing the minimum $2$-respecting cut, for a given tree $T$. That is, given a fixed spanning tree $T$ of a graph $G$, compute the minimum cut in $G$ among those that cut at most $2$ edges of $T$. %Eventually, we would be doing this for $\poly(\log n)$ given trees, provided by the tree packing, and take the overall minimum.

To treat the minimum $2$-respecting cut problem for the given tree $T$, we break it into simpler special cases. The general algorithm will be achieved by a clean and modular combination of these cases. \medskip

\noindent\paragraph{Path-to-path 2-respecting min-cut.}\nquad(\Cref{sec:path-to-path-cut})
First, we consider an important sub-case of the \emph{path-to-path 2-respecting min-cut}: the case when the tree $T \subseteq G$ happens to be composed of exactly two paths $P$ and $Q$ along with a common root connecting them (see \Cref{fig:path-to-path}). Our goal is to find the $\min \Cut(e, f)$ over all pairs $e, f \in E(P) \times E(Q)$ (i.e., the edges are on different paths), where $\Cut(e, f)$ is the sum of weights of edges of $G$ which cross the cut determined by $(e, f)$ (i.e., all edges with endpoints $u, v$ such that the unique $T$-path between $u, v$ crosses exactly one of $\{e, f\}$). Several notable ideas go into designing an algorithm for this problem:

\begin{itemize} 
\item \emph{General recursive idea and the Monge property.} (Observed by Mukhopadhyay and Nanongkai~\cite{mukhopadhyay2020weighted}.) The main idea is to import the state-of-the-art centralized techniques into the distributed setting with the help of the Minor-Aggregation model. Specifically, we first fix $e_a$ to be the midpoint edge of $P$ (i.e., $a := \lfloor |P| / 2 \rfloor)$ and let $f_b$ be the \textbf{best response} to $f_a$, meaning the edge $f_b \in E(Q)$ that minimizes $\Cut(e_a, f_b)$. Then, either $(e_a, f_b)$ is the pair that minimizes the 2-respecting cut, or the minimizing pair can be found on either $P_{\text{up}} := \{e_1, \ldots, e_{a-1}\} \times Q_{\text{up}} :=\{f_1, \ldots, f_{b-1}\}$ ($e_1, f_1$ are connected to the root) or $P_{\text{down}} := \{e_{a+1}, \ldots, e_{|P|}\} \times Q_{\text{down}} := \{f_{b+1}, \ldots, f_{|Q|}\}$. This property, i.e., that the minimizing pair is either completely on one side or completely on the other side of $(e_a, f_b)$, is the so-called Monge property. Due to this property, we can issue two simultaneous recursive calls on $P_{\text{up}} \times Q_{\text{up}}$ and $P_{\text{down}} \times Q_{\text{down}}$ and return the best result found. Note that a parallel implementation of this idea has $\tl{O}(1)$ recursion depth and can be implemented in near-linear work.
  
\item \emph{Private cut-equivalent graphs.} One issue afflicting the above idea in the distributed setting is that the recursive call on, say, $P_{\text{up}} \times Q_{\text{up}}$ requires information private to $P_{\text{down}} \times Q_{\text{down}}$: an edge strictly between the latter affects the answer of the former. 
%try to use the same edge of $G$ for communication, thereby causing a congestion blowup on such an edge, leading to significant slowdowns.
To prevent this, we construct \emph{private} and \emph{cut-equivalent} graphs $G_{\text{up}}$ and $G_{\text{down}}$ that are (1) private, in the sense that the recursions can freely use them as well as guaranteeing that (2) $\Cut(e \in E(P\up), f \in E(Q\up))$ is the same with respect to $G$ and $G\up$. This is achieved by replacing the top-most and bottom-most edges of $P\up$ and $Q\up$ with virtual nodes (as well as the root), which are both private to the recursion and allow us to insert additional edges to achieve cut equivalency. For example, an edge $\{a \in V(P\down), b \in V(Q\up)\}$ is replaced in $G\up$ with an edge between the bottom (virtual) node of $P\down$ and $b$, making it private and making its contribution to all 2-respecting cuts equivalent in the recursive call on $G\up$ as it would have been if considering $G$. Other types of edges and the recursive call on $G\down$ are analogous.

\item \emph{Avoiding simulation cascade (using separability).} Another issue with the above idea of private-but-virtual graphs is that each recursive call is performed on a virtual graph (albeit, with a small number of virtual nodes). This has to be ultimately converted to an algorithm without virtual nodes. For instance, one idea is to naively call the recursive algorithm on, say, (the virtual graph) $G\up$ and then remove the virtual nodes using simulation (which introduces a small multiplicative overhead). However, this would yield a runtime explosion as every level of the recursion would introduce a cascading multiplicative overhead to the computation, making the final runtime polynomial (the desired runtime is polylogarithmic). The solution, however, might seem simple but is essential. Consider, say, the sub-instance $P\up \times Q\up$ on $G\up$. We want to remove the virtual nodes before the recursive call returns so that the returned call only performs work on the underlying graph, removing any need for cascading simulation of virtual nodes. This ``de-virtualization'', however, can only be performed in $G\up$ if $G\up - \mathrm{Virt}$ (minus its virtual nodes) is connected. If this is the case, we can resolve the issue as explained. If it is not connected, however, this forces a trivial structure called \emph{separability} on the sub-instance which can be solved without recursing. Specifically, we show that $\Cut(e, f)$ can be \emph{separated}, i.e., written as $\Cut(e, f) = F_P(e) + F_Q(f)$ for some functions $F_P, F_Q$. In this case, separate minimizations of both sides lead to the correct result.
  %replace edges that are not fully contained in $P\up \times Q\up$ but influence the result of the recursion to be inserted within the private graph ????. The issue is analogous for $P\down \times Q\down$.
\end{itemize}

\noindent\paragraph{Star 2-respecting min-cut.}\nquad(\Cref{sec:star-cut}) Next, we use the path-to-path algorithm to build an algorithm in which the tree $T \subseteq G$ is exactly composed of $k$ paths $P_1, \ldots, P_k$ and a common root that connects to the top of each path (see \Cref{fig:star}). The goal is to find the minimum 2-respecting cut $\Cut(e, f)$ where $e \in E(P_i)$ and $f \in E(P_j)$ are two edges on different paths $i \neq j$. Several notable ideas go into designing an algorithm for this problem:

\begin{itemize}
\item \emph{Path interest.} (Introduced by Mukhopadhyay and Nanongkai~\cite{mukhopadhyay2020weighted}.) We say a non-tree edge $\{u, v\} \in E(G)$ \textbf{covers} a tree-edge $e$ if the unique path in $T$ between $u$ and $v$ contains $e$. We say that a path $P_i$ is \emph{ interested} in a path $P_j$ if there exist edges $e \in E(P_i), f \in E(P_j)$ such that at least half of the edges covering $e$ also cover both $e$ and $f$ (counting weight as the multiplicity, see \Cref{lemma:interesting-path-vs-cov}). If the pair of edges that determine the optimum 2-respecting min-cut lie on paths $P_i$ and $P_j$, then $P_i$ and $P_j$ must be mutually interested in each other (\Cref{lemma:best-cut-has-interest}). Therefore, the general idea for the star algorithm will be to examine all mutually-interested pairs of paths using the path-to-path oracle. An important property that enables solving the star instance is that each path is interested in at most $O(\log n)$ other paths.

\item \emph{Interest lists and cross-edges.} We now describe how to efficiently compute for each path $P_i$ a list of paths that $P_i$ is interested in. On a technical level, each path-edge $e \in E(P_i)$ needs to find the set of edges $f \in E(P_j)$ such that the majority of non-tree edges covering $e$ also cover $f$. It is immediate that, for each fixed $e$, all edges $f$ (it any) lie on a single path $P_j$ in which case  $P_i$ is interested in $P_j$. For a fixed edge $e \in E(P_i)$, this corresponds to picking a majority element of a sequence, where each (non-tree) edge $f \in E(G)$ between $P_i$ and $P_j$ contributes $w(f)$ weight to $P_j$. This majority operation, however, can be performed using deterministic \emph{heavy-hitter sketches}, which gracefully fit within the framework of aggregation operations. Therefore, we can use the newly developed deterministic subtree sum operation with the heavy-hitter aggregator to find the majority element for each edge $e$, indicating path interest. Furthermore, as $e \in E(P_i)$ is ``moved across'' the path $P_i$, there can be at most $O(\log n)$ other paths that $P_i$ is interested in (\Cref{lemma:path-weak-interest-cardinality}). Therefore, we find the union of all found (almost) majority elements in each path, as there can be at most $\tl{O}(1)$ of them. 

%needs to consider all (non-tree) edges $f \in E(G)$ covering it. For each such $f$, consider which other path $P_j$ is $f$ covering (there can be at most one such other path $P_j$). Then, $e$ needs to pick out the majority element $P_j$ (counting for weight of $f$ as the multiplicity of $P_j$).

However, there is an issue plaguing this approach: if one simply considers all edges $f \in E(G)$ covering a path-edge $e \in E(P_i)$ and is looking for the majority element using the subtree sum operation, they would also need to support the ``remove'' operation in the heavy-hitter sketch since some edges considered throughout the subtree of a node $v$ should not be considered at its parent node. However, the heavy-hitter sketch does not support this. We get around this by slightly changing the definition of interest to only consider \emph{cross-edges} (edges going from one path to another), which do not require removals. We show that all the important results hold even if one ignores all other types of edges.
  
\item \emph{Interest graph.} Consider the logical graph where each node represents a different path $P_i$ and there is an edge $\{P_i, P_j\}$ if and only if $P_i$ and $P_j$ are mutually  interested in each other. Moreover, we can simulate an arbitrary Minor-Aggregation algorithm on the interest graph by contracting away all path edges since any two mutually-interested paths must have an edge between them. Moreover, since each path is interested in at most $O(\log n)$ other paths, the maximum degree of the interest graph is $\Delta := O(\log n)$. This implies that we can also simulate arbitrary CONGEST algorithms on the interest graph (i.e., non-aggregation based) with a multiplicative $O(\Delta) = \tl{O}(1)$ blowup (\Cref{lemma:hl-interest-simulation}).

\item \emph{Edge coloring of the interested graph.} We find the smallest 2-respecting cut among all pairs of mutually-interested paths by first computing an edge coloring of the interest graph. To this end, we can simulate the deterministic CONGEST algorithm of Panconesi and Rizzi~\cite{panconesi2001some} on the interest graph that colors the interest graph into $O(\Delta) = \tl{O}(1)$ colors. Then, we iteratively consider each color class in isolation. Within each class, all pairs of matched paths are node disjoint, hence we can use the previously-developed path-to-path 2-respecting min-cut algorithm to find the optimum solution.
\end{itemize}

\noindent\paragraph{Between-subtree 2-respecting min-cut.}\nquad(\Cref{sec:2-respecting-between-subtree}) We now use the star algorithm to build a between-subtree 2-respecting cut algorithm, in which the tree $T \subseteq G$ is exactly composed of $k$ subtrees $T_1, \ldots, T_k$ and a common root that connects to the top of each subtree (see \Cref{fig:between-subtree}). The goal is to find the minimum 2-respecting cut $\Cut(e, f)$ where $e$ and $f$ are two edges in different subtrees. Several notable ideas go into designing an algorithm for this problem:
\begin{itemize}
\item \emph{Pairwise coloring.} Our first idea is to reduce the problem for general $k$ to the case when $k = 2$. Suppose the optimum 2-respecting cut $(e^*, f^*)$ is contained in subtrees $e^* \in E(T_{i^*})$ and $f^* \in E(T_{j^*})$. We will construct a \emph{pairwise coloring} $\{f_1, \ldots, f_\chi \}$, i.e., a small collection of \emph{color assignments} $f_i : [k] \to \{ \mathrm{red}, \mathrm{blue} \}$ such that each pair of subtrees $T_i, T_j$ is assigned a different color in at least one color assignment. It is a folklore result that there exists such an assignment with $\chi = O(\log n)$ (e.g., consider the $O(\log n)$ different bits of the subtree IDs). After constructing such a collection of colorings, we iterate over each color assignment, and for each assignment, merge all the roots of all subtrees colored $\mathrm{red}$ and all subtrees colored $\mathrm{blue}$. This reduces the problem to the $k=2$ case.
  
\item \emph{Heavy-light decomposition.} We now reduce the $k = 2$ problem to the (solved) star case. First, we construct a heavy-light decomposition of both subtrees, in which each edge is assigned a label ``heavy'' or ``light'' such that each root-to-leaf path has at most $O(\log n)$ light edges. We define an HL-depth of an edge $e$ to be the number of light edges on the root-to-$e$ path. Now, suppose the optimum 2-respecting cut $(e^*, f^*)$ has $d_1^* := \HLdepth(e^*)$ and $d_2^* := \HLdepth(f^*)$. Since $d_1, d_2 = O(\log n)$, we guess the correct $d_1^*$ and $d_2^*$ by testing all possible combinations. For each guess, contract all edges $e$ in $T_1$ with $\HLdepth(e) \neq d_1^*$ and all edges $e$ in $T_2$ with $\HLdepth(e) \neq d_2^*$. This reduces the question to exactly the star case (see \Cref{fig:subtree-reduction}). %\alert{PICTURE!}
\end{itemize}

\noindent\paragraph{Final step: 2-respecting general cut.}\nquad(\Cref{sec:general-2-resp-cut}) Finally, we solve the general 2-respecting min-cut, in which we are given a spanning tree $T$ of a weighted graph $G$ and the goal is to find $\min_{e \in E(T), f \in E(T)} \Cut(e, f)$. Several notable ideas go into designing an algorithm for this problem:
\begin{itemize}
\item \emph{The general recursive idea and the centroid decomposition.} It is a well-known folklore result that each tree $T$ has a centroid node $c \in V(T)$ such that all connected components of $T - v$ have at most $|V(T)| / 2$ nodes. We will solve the general case by first finding the centroid $c$ of our tree $T$, which can be performed using the subtree sum operation. Now, let us denote the pair of edges defining some 2-respecting min-cut by $(e^*, f^*) \in E(T) \times E(T)$, and suppose we denote the maximal connected subtrees of $T - c$ by $T_1, \ldots, T_k$. Then, the pair $e^*, f^*$ can either be (1) in two different subtrees $T_{i^*}, T_{j^*}$, or (2) in the same subtree $T_{i^*}$. For case (1) we simply need to call the 2-respecting between-subtree cut algorithm on $T_1, \ldots, T_k$; for case (2) we will use recursion on each one of the (node disjoint) subtrees $T_i$, allowing us to schedule all recursive calls simultaneously.

\item \emph{Cut-equivalent subtrees via virtual nodes.} One immediate issue breaking a naive implementation of the above recursive idea is that each recursive call on, say, $T_i$ needs to have a private copy of edges $E(G)$ which are used to calculate the values of (2-respecting) cuts. The issue seems essential: cut values completely within $T_i$ can, after a few levels of recursion, depend on edges whose both endpoints are in unrelated recursive calls. While this is not as big of a problem in the parallel or centralized settings, as one can build global data structures shared across recursive calls (e.g., as done in \cite{gawrychowski_et_al:LIPIcs:2020:12464}), this is a fundamental issue in the distributed setting. However, with our extensions to the Minor-Aggregation model, we can tackle this using virtual nodes, which can be arbitrarily connected even to non-virtual nodes. Upon finding the centroid $c$, we attach a (private) ``virtual centroid'' $c_i$ to each subtree $T_i$. Furthermore, if there is an edge $\{u, v \} = e \in E(G)$ crossing between subtrees from, say, $u \in V(T_i)$ to $v \in V(T_j)$, then we will add two virtual edges $\{u, c_i \}$ and $\{ v, c_j \}$, both of the same weight as $w(e)$ (see \Cref{fig:centroid-decomposition}). It is easy to show that such a transformation preserves all 2-respecting cuts within all subtrees $T_i$ (modified with the virtual centroid and corresponding edges), and all of these subgraphs are private to their own recursive call.

\item \emph{Avoiding simulation cascade.} Naively recursing on each $T_i + c_i$ and eliminating the virtual node by simulation leads to a simulation cascade, preventing us from achieving the desired runtime. However, we can mitigate this in the following way. Consider some particular recursive call, which happens to be run on some tree $T$. In parent recursive calls several virtual nodes $\mr{Virt} \subseteq V(T)$ were introduced, one per level of the recursion. However, our algorithm has the immediate property that $T - \mr{Virt}$ (all virtual nodes and their adjacent edges removed) is connected. Furthermore, due to the choice of the centroid as the pivoting node, the depth of the recursion is $O(\log n)$, giving us a bound that $|\mr{Virt}| \le O(\log n)$. Therefore, inside the recursive call before returning, we eliminate all the virtual nodes $\mr{Virt}$ by simulating the algorithm on $T - \mr{Virt}$, which is a connected subgraph of the underlying communication network. Hence, no additional (or cascading) virtual node elimination needs to happen after the recursive call returns.
  
\end{itemize}

\section{Preliminaries}
\label{sec:prelim}
\noindent\textbf{Basic Notations.} We define $[k] := \{1, 2, \ldots, k\}$. By $A \sqcup B$, we mean the disjoint union of $A$ and $B$. 
\medskip

\noindent\textbf{Graphs.} An undirected graph $G$ is composed of a node set $V(G)$ and an edge set $E(G)$. We often work with weighted graphs, in which case each edge $e$ is assigned a weight $w(e)$ that is polynomial in the number of nodes, i.e., $w(e) \in [\poly(n)]$, where we use the usual notation of $n := |V(G)|$ used throughout this paper. We use the $G[P]$ to denote the subgraph induced by vertices in the set $P\subset V(G)$. Given a subset $D \subseteq V(G)$, we denote with $G - D$ the subgraph resulting from removing all nodes $D$ and their incident edges from $G$. To simplify notation, we use $G - v$ instead of $G - \{v\}$ when $v \in V(G)$ is a node.

\medskip

\noindent\textbf{Rooted trees.} Let $T = (V(T), E(T))$ be a tree with a specially designated vertex $r$ called the root. The edges $E(T)$ are called \textbf{tree-edges}. If $\{u, v\}$ is an edge in a rooted tree and $u$ is closer to the root, then $u := \parent(v)$ is the \textbf{parent} of $v$ and $v$ is a \textbf{child} of $u$. Alternatively, given an edge $e = \{ u, \parent(u) \}$, we write $\mytop(e) = \parent(u)$ for the ``top endpoint'' (i.e., closer to the root) and $\bottom(e) = u$ for the ``bottom endpoint''. A node is a \textbf{leaf} if it has no children. A node $u$ is an \textbf{ancestor} or $v$ if the root-to-$v$ path contains $u$. The set of all ancestors of $v$ is denoted by $\anc(v)$ ($v$ included). Similarly, $u$ is a \textbf{descendant} of $v$ if $v$ is an ancestor of $u$, and we write this as $u \in \desc(v)$. Note that $\desc(u) \ni u \in \anc(u)$. The \textbf{depth} of $v$, denoted by $\depth(v)$, is the (hop-)distance between the root and $v$. The \textbf{subtree} at $u$, denoted by $\subtree(u)$, is the induced subgraph $T[\desc(u)]$. The \textbf{lowest common ancestor (LCA)} of nodes $u$ and $v$ is the (unique) node $\LCA(u, v) = w$ with the largest depth such that both $u$ and $v$ are in the subtree of $w$. A path $(v_1, v_2, \ldots, v_k)$ (where $v_i \in V(T)$) is \textbf{descending} if $v_{i+1}$ is a child of $v_i$ for all $i \in \{1, \ldots, k-1\}$ (i.e., it is a subpath of a root-to-leaf path).
%An (possibly non-tree) edge $e = \{u, v\}$ is a \textbf{back-edge}\alert{IS THIS USED?} if the $\LCA(u, v)$ is either $u$ or $v$; otherwise it is called a cross-edge\alert{IS THIS USED?}.

% \textbf{Paths in rooted trees.} . We say two edges $e, f$ are \textbf{orthogonal} if there is no root-to-leaf path containing both of them. Two paths $P_1, P_2$ are called orthogonal if all pairs of edges $E(P_1) \times E(P_2)$ are orthogonal.
% \text{Task specification.} We often specify tasks in the following form.
% \PaTaskNIO{task:example-task}{Example task.}{This field specifies the underlying network that is used for communication. The nodes initially know the network graph, stored in a distributed fashion (see below).}{This field specifies what is initially known by the nodes. Unless otherwise specified, the information is stored in a distributed fashion (see below).}{Upon successful termination of this task, the nodes learn the information specified in this field. Unless otherwise specified, the information is stored in a distributed fashion (see below).}

\subsection{Heavy-light decomposition}

In this section, we review the well-known \emph{heavy-light decomposition} that decomposes a tree $T$ into ``HL-paths'' such that each root-to-leaf path in $T$ can be composed into at most $O(\log n)$ different HL-paths (e.g., see Lemma 5 of \cite{bhardwaj2019simple}).
\begin{definition}
  Given a rooted tree $T$, a heavy-light decomposition is a labeling of edges of $T$ where each $e \in E(T)$ is assigned a label of either \textbf{heavy} or \textbf{light} in the following way. Let $S(v) = |\desc(u)|$ be the number of descendants of $v$. For each non-leaf node $u$ find its child $v$ that maximizes $S(u)$ and label the edge $\{u, v\}$ ``heavy'' (breaking ties arbitrarily); all other edges are \textbf{light}.
\end{definition}
\begin{fact}\label{lemma:hl-has-logn-light-edges}
  Given any heavy-light decomposition of a rooted tree, every root-to-leaf path has at most $O(\log |V(T)|)$ light edges.
\end{fact}

\noindent We define some terminology used throughout the paper:
\begin{itemize}\setlength\itemsep{0em}
\item \textbf{HL-depth.} The HL-depth of a node $v$ is the number of light edges on the root-to-$v$ path. The HL-depth of a tree-edge $e$ is the HL-depth of its node farther away from the root, i.e., $\HLdepth(e) = \HLdepth(\bottom(e))$.

\item \textbf{HL-path.} An HL-path is a maximal ancestor-to-descendant path in $T$ where all edges have equal HL-depths. Note: an HL-path is a proper path in the tree and the \emph{edges} (but not the nodes due to the endpoints) of a tree can be (disjointly and completely) partitioned into HL-paths. Specifically, an HL-path includes its top-most light edge.

\item \textbf{HL-info.} The HL-info of a node $v$ consists of (1) the $T$-depth of $v$, and (2) the list $L_v$, where for each light edge $e$ on the root-to-$v$ path we store the $T$-depth and ID for both of $e$'s endpoints.
\end{itemize}

A simple but useful property of heavy-light decompositions is that they can be used as LCA labeling schemes, formalized below.
\begin{fact}\label{lemma:lca-labels}
  There exists a function that takes only the HL-infos of any two nodes $u$ and $v$, and computes the (ID and depth of the) LCA of $u$ and $v$.
  %Specifically, either $u$ is an ancestor or $v$, or $v$ is an ancestor of $u$, or $\LCA(u,v)$ is the top node of an edge in $L_u \cup L_v$, where $L_u$ contains the light edges on the root-to-$u$ path.
\end{fact}

% \begin{lemma}\label{lemma:hl-info-small-size}
%   The number of bits necessary to describe HL-info for a single node is $\tO(1)$.
% \end{lemma}

\subsection{Min-cut specifics: cut and cover values}\label{sec:mincut-prelims}

Following Dory et al.~\cite{dory2021distributed}, we formalize the notions of cut values, cover values, and 1-/2-respecting cuts that are used throughout this paper. In the following, suppose $T \subseteq G$ is a spanning tree of a weighted graph $G$:
\begin{itemize}\setlength\itemsep{0em}
\item \textbf{Cut values.} For $e, f \in E(T)$ we define the \textbf{cut value} $\Cut_{T, G}(e, f)$ as the the sum of weights of all edges $g \in E(G)$ that cross the unique cut which cuts exactly $\{e, f\}$ among all tree edges $E(T)$. In other words, the sum of weights of all edges $\{u, v\} \in E(G)$ whose unique $T$-path between $u$ and $v$ contains exactly one of $\{e, f\}$. $\Cut_{T, E}(e)$ is defined analogously: sum of weights of all edges $\{u, v\} \in E(G)$ whose unique $T$-path between $u$ and $ v$ contains $e$. 

\item \textbf{Cover values.} For $e, f \in E(T)$ we define the \textbf{cover value} $\Cov_{T, G}(e, f)$ as the sum of weights of all edges $\{u, v\} \in E(G)$ such that the unique $T$-path between $u$ and $v$ covers both $e$ and $f$. We also define $\Cov_{T, G}(e) := \Cov_{T, G}(e, e)$.

\item \textbf{1- and 2-respecting cuts.} The cuts corresponding to $\Cut_{T, G}(e)$ and $\Cut_{T, G}(e, f)$ are called 1-respecting and 2-respecting cuts (with respect to a tree $T$), respectively. The 1-respecting and 2-respecting min-cut values are defined as $\min_{e \in E(T)} \Cut_{T, G}(e)$ and $\min_{e \in E(T), f \in E(T)} \Cut_{T, G}(e, f)$, respectively.
\end{itemize}

We often drop the subscript when $G$ and $T$ are apparent from the context. We point out a few useful observations about these values. The first one is immediate, while the second is an important observation of Mukhopadhya and Nanongkai~\cite{mukhopadhyay2020weighted}.
\begin{fact}\label{lemma:2-cut-vs-cov}  
  Given a spanning tree $T \subseteq G$, for all $e, f \in E(T)$ we have that $\Cut(e) = \Cov(e)$ and $\Cut(e, f) = \Cov(e) + \Cov(f) - 2 \Cov(e, f)$.
\end{fact}
\begin{fact}\label{lemma:interesting-path-vs-cov}
  Given a spanning tree $T \subseteq G$, if $\Cut(e, f)$ is smaller than any 1-respecting cut, then $\Cov(e, f) > \Cov(e) / 2$.
\end{fact}
\begin{proof}
  Since $\Cut(e, f)$ is smaller than any 1-respecting cut, we have that $\Cut(e, f) < \Cut(f)$. Using, $\Cut(e, f) = \Cov(e) + \Cov(f) - 2\Cov(e, f)$ (\Cref{lemma:2-cut-vs-cov}), we get $\Cov(e) + \Cov(f) - 2 \Cov(e, f) = \Cut(e, f) < \Cov(f)$. Therefore, we get $\Cut(e) < 2 \Cov(e, f)$.
\end{proof}

% \subsection{Low-congestion shortcuts}

% \alert{I don't think shortcuts are needed at all in this paper. Maybe only shortcut quality}
% \alert{define shortcut quality $\shortcutQuality{G}$}
% \alert{define ``shortcuts are inapproximable''}

% \begin{restatable}{definition}{shortcutFormal}\label{def:shortcutFormal}[Shortcut quality, copied from \cite{haeupler2020shortcuts}]
%   A shortcut for parts $(P_1, \ldots, P_k)$ is $( H_1, \ldots, H_k )$, where $H_i$ is a subset of edges of $G$. The shortcut has \textbf{dilation} $d$ and \textbf{congestion} $c$ if (1) the diameter of each $G[P_i] \cup H_i$ is at most $d$ (i.e., between every $u, v\in P_i$ there exists a path of length at most $d$ using edges of $G[P_i] \cup H_i$), and (2) each edge $e$ is in at most $c$ different sets $H_i$. The \textbf{quality} of the shortcut is $Q = c+d$.
% \end{restatable}
% %\end{wrapper}

% %% COPIED, make this nicer!
% Classic routing results by Leighton, Maggs and Rao~\cite{leighton1994packet} show that such \alert{rephrase, partwise aggregation not defined yet!!} a shortcut allows for the part-wise aggregation to be solved in $\tl{O}(c+d)$ rounds, even distributedly. This motivates the definition of quality. We say a network topology $G$ admits low-congestion shortcuts of quality $Q$ if a shortcut with quality $Q$ exists for \emph{every} partition into disjoint connected parts of $G$'s nodes. Denoting the minimum such quality $Q$ by $\shortcutQuality{G}$, this approach leads to algorithms whose running times are parameterized by $\shortcutQuality{G}$.

\subsection{The distributed Minor-Aggregation model}\label{sec:prelim-minor-aggregation}

In this section, we describe the Minor-Aggregation model. We first define aggregation operators, then give a formal description of the Minor-Aggregation model as defined in \cite{goranci2022universally}; we extend the model in \Cref{sec:extending-minor-aggregation}. Since we are designing distributed graph algorithms, throughout this paper we assume some underlying undirected graph $G$ called the \emph{communication network} or the \emph{network topology} over which we run our distributed algorithms. We will denote with $n := |V(G)|$ the number of nodes of the network. Moreover, we will typically make use of the $\tl{O}$-notation which hides $\poly(\log n)$ factors as the theory of low-congestion shortcuts (and by extension, the Minor-Aggregation model) is generally tight only up to polylogarithmic factors; which is still a significant improvement over the polynomial overhead factors previously present in pre-shortcut algorithms.

\subsubsection{Aggregation operators}

Aggregations are simple functions $\bigoplus$ (e.g., sum or max) that produce an aggregate value $\bigoplus_{i=1}^k x_i$ from a sequence of values; we formalize the notion below.
\begin{definition}[Aggregation operator]
  An aggregation operator $\bigoplus$ takes two $B$-bit messages $m_1, m_2$ for some $B = \tl{O}(1)$ and combines them into a new $B$-bit message $m_1 \bigoplus m_2$. Furthermore, given $k$ messages $m_1, \ldots, m_k$, their $\bigoplus$-aggregate $\bigoplus_{i=1}^k m_i$ is the message resulting from an arbitrary sequence of operations that takes any two messages $m', m''$, deletes them from the sequence, and replaces them with a single $m' \bigoplus m''$, repeating until a single message remains.
\end{definition}
Most commonly, aggregations will be commutative and associative (e.g., sum or max), which makes the value $\bigoplus_{i=1}^n m_i$ unique. For example, the sum-aggregation of $x_1, x_2, \ldots, x_k$ is simply $x_1 + x_2 + \ldots + x_k$. However, it is often very convenient to gain additional flexibility by allowing more general aggregation operators where the output might depend on the execution sequence. This allows us to use any \emph{mergeable $B$-bit sketch}~\cite{agarwal2013mergeable} as an aggregation operator, giving us a way of computing statistics such as approximate heavy hitters, or approximate quantiles, even in the deterministic setting. The following example illustrates this point on the well-known ``heavy hitters'' sketching algorithm described by Misra and Gries~\cite{misra1982finding}.
% Moreover, in the randomized setting, this enables moment estimations (since linear sketches are trivially mergeable), random samples of data, 
%We illustrate this point by showing the well-known ``heavy hitters'' sketching algorithm can be converted into an aggregation operator. 

\begin{example}[Deterministic Approximate Heavy Hitters]\label{example:heavy-hitters} % deterministic! % https://web.stanford.edu/class/cs168/l/l2.pdf
  Given a list of $k \le \exp(\poly(\log n))$ objects $e_1, \ldots, e_k$ from a $\exp(\poly(\log n))$-sized universe with multiplicities $w(e_i)$, we define the \emph{frequency} $f(x)$ of an object $x$ as the sum of weights $f(x) := \sum_{i} \1{e_i = x} \cdot w(e_i)$ across all appearances of $x$ in the list. Let $W := \sum_{i=1}^k w(e_i)$ be the total weight. For any integer $h > 0$, there exists a $\tl{O}(h)$-bit deterministic aggregation operator $\bigoplus$, where $\bigoplus_{i=1}^k e_i$ returns a list of $h$ elements (and their estimated frequencies), such that (1) each object $x$ with $f(x) > \frac{2}{h} W$ is included in the list, and (2) no object $e$ with $f(e) \le \frac{1}{h} W$ is included in the list.
\end{example}
% \begin{proof}
%   \alert{fill this in}
% \end{proof}

% %and the operator randomness. For example, an appropriately defined $\bigoplus$ enables finding a random sample of a sequence.
% we okasdwe okasdwe okasd
% \begin{example}[Deterministic Heavy Hitters]\label{example:heavy-hitters} % deterministic! % https://web.stanford.edu/class/cs168/l/l2.pdf
%   Suppose we are given a list $L = ( (e_i, w_i \in \mathbb{Z}_{\ge 0}) )_{i=1}^k$ composed of value, weight pairs (both $\tl{O}(1)$-bit). We say the ``sum of weights of an element $e$'' is $\sum_{i=1}^k w_i \1{e_i = e}$ and the ``total weight'' is $\sum_{i=1}^k w_i$. For any integer $h \ge 2$ the following holds. There exists an aggregation operator $\bigoplus$ which operates on messages of size $\tl{O}(h)$ which outputs a list of elements such that (1) each element $e$ whose sum of weights is more than $\frac{2}{h}$-fraction of the total weight is included in the list, and (2) elements whose sum of weights is at most $\frac{1}{h}$-fraction of the total weight are not included in the list.
% \end{example}
% \begin{proof}
%   \alert{fill this in}
% \end{proof}

\subsubsection{Minor-Aggregation model: an interface for distributed algorithms}

In this section, we formally define the Minor-Aggregation model introduced in \cite{goranci2022universally}, a powerful interface that facilitates simple design of ultra-fast distributed algorithms. Algorithms in the Minor-Aggregation model can be compiled-down to the standard CONGEST message-passing settings~\cite{goranci2022universally} (see \Cref{sec:det-ops-and-simulation}) or to parallel settings (e.g., as in \cite{rozhon2022undirected}).
% \smallskip
% \noindent\textbf{CONGEST.} We are given an $n$-node undirected graph $G = (V, E)$ called the ``communication network''. The vertices are called nodes and they are individual computational units (i.e., have their own processor and private memory). Communication between the nodes occurs in synchronous rounds. In each round, each pair of nodes adjacent in $G$ exchange an $O(\log n)$-bit message. Nodes perform arbitrary computation between rounds. Initially, nodes only know their unique $O(\log n)$-bit ID and the IDs of adjacent nodes. At the end, each node should output its own part of the output, e.g., which of its edges are in the computed exact minimum-cut. In the case of minimum cut, we require that all nodes also know the size of the minimum cut.

%\textbf{Agreement and knowledge sharing.} We say a set of computational units (e.g., nodes or edges) \emph{agree} on some value $x$ if, when evaluating $x$, all of them output the same value. 

\begin{definition}[Distributed Minor-Aggregation Model]\label{def:aggregation-congest}
  We are given a connected undirected graph $G = (V, E)$. Both nodes and edges are individual computational units (i.e., have their own processor and private memory). Communication occurs in synchronous rounds. Initially, nodes only know their unique $\tl{O}(1)$-bit ID and edges know the IDs of their endpoints. Each round consists of the following three steps (in that order).
  \begin{itemize}    
  \item \textbf{Contraction step.} Each edge $e$ chooses a value $c_e = \{\bot, \top\}$. This defines a new \emph{minor network} $G' = (V', E')$ constructed as $G' = G / \{ e : c_e = \top \}$, i.e., by contracting all edges with $c_e = \top$ and self-loops removed. Vertices $V'$ of $G'$ are called supernodes, and we identify supernodes with the subset of nodes $V$ it consists of, i.e., if $s \in V'$ then $s \subseteq V$.

  \item \textbf{Consensus step.} Each node $v \in V$ chooses a $\tl{O}(1)$-bit value $x_v$. For each supernode $s \in V'$, we define $y_s := \bigoplus_{v \in s} x_v$, where $\bigoplus$ is any pre-defined aggregation operator. All nodes $v \in s$ learn $y_s$.

  \item \textbf{Aggregation step.} Each edge $e \in E'$, connecting supernodes $a \in V'$ and $b \in V'$, learns $y_a$ and $y_b$ and chooses two $\tl{O}(1)$-bit values $z_{e, a}, z_{e, b}$ (i.e., one value for each endpoint). For each supernode $s \in V'$, we define an aggregate of its incident edges in $E'$, namely $\bigotimes_{e \in \text{incidentEdges(s)}} z_{e, s}$ where $\bigotimes$ is some pre-defined aggregation operator. All nodes $v \in s$ learn the aggregate value (they learn the same aggregate value, a non-trivial assertion if there are many valid aggregates).
  \end{itemize}
\end{definition}

\noindent\textbf{Operating on minors.} A particularly appealing feature of the Minor-Aggregation model is that the framework immediately allows any black-box algorithm to run on a minor of a graph rather than on the original graph (due to contractions). One notable difference from CONGEST that makes this possible is that nodes in the Minor-Aggregation model do not have a list of their neighbors\footnote{Moreover, it is not hard to see that a model which allows contractions and gives nodes a list of their neighbors cannot be simulated in CONGEST with $o(n)$ round blowup, even for graphs of small diameter.}. The following corollary is immediate from the definition.
\begin{corollary}\label{corollary:computation-on-minors}
  Any $\tau$-round Minor-Aggregation algorithm on a minor $G' = G / F$ of $G = (V, E)$ can be simulated via a $\tau$-round Minor-Aggregation algorithm on $G$. Initially, each edge $e \in E$ needs to know whether $e \in F$ or not. Upon termination, each node $v$ in $G$ learns all the information that the $G'$-supernode $v$ was contained in learned.
\end{corollary}

\noindent \textbf{Node-disjoint scheduling.} We can run simultaneous algorithms on connected node-disjoint subgraphs.

\begin{corollary}\label{lemma:node-disjoint-scheduling}
  Let $G$ be an undirected graph. Given any $\tau$-round Minor-Aggregation algorithms $A_1, \ldots, A_k$ running on node-disjoint and connected subgraphs (of $G$) $H_1, H_2, \ldots, H_k$, we can run $A_1, \ldots, A_k$ simultaneously within a $\tau$-round Minor-Aggregation algorithm $A$ that runs on $G$.
\end{corollary}

\noindent\textbf{Distributed storage.} Distributed algorithms often require the computation of various global structures like spanning trees. Storing such structures on any single node is prohibitively expensive (would require a linear number of rounds on some graphs). Therefore, such global structures are stored locally---each node remembers its own part. We specify how different structures are stored below (a notable missing entry is the storage of virtual graphs, which is specified in \Cref{sec:virtual-nodes}).
\begin{itemize}  
\item We distributedly store a \textbf{node vector} $x \in \R^V$ by storing the value $x_v$ in the node $v \in V$. Similarly, given an \textbf{edge vector} $x \in \R^E$ we store the value $x_e$ in the edge $e$ (we remind the readers that edges are computational units in the Minor-Aggregation model).

\item We distributedly store a \textbf{subgraph} $H \subseteq G$ of the communication network $G$ (where $G$ is the communication network) by storing distributedly storing the indicator node and indicator edge vectors $x_v := \1{v \in V(H)} \in \{0, 1\}^V$ and $y_e = \1{e \in E(H)} \in \{0, 1\}^E$.

\item A \textbf{rooted tree} $T = (V, E_T)$ with a root $r \in V$ is distributedly stored if (1) the unoriented version of $T$ is stored as a subgraph, and (2) each edge $e \in E_T$ knows which endpoint is closer to the root $r$ in $T$.

\end{itemize}

\noindent\textbf{Input/Output.} When stating a result about an algorithm that requires input and produces output, we implicitly mean that all inputs are assumed to be distributedly stored before the algorithm is being run, and the output will be distributedly stored upon termination. For example, for the 2-respecting cut problem, we expect the weights of $E(G)$ to be stored as an edge vector and the spanning tree $T$ to be stored as a subgraph.

\medskip

\noindent\textbf{Simulation in CONGEST.} An algorithm in the Minor-Aggregation model can be efficiently simulated in the standard CONGEST model if one can efficiently construct shortcuts~\cite{goranci2022universally}. We will formally state the deterministic simulation result in \Cref{sec:det-ops-and-simulation}.

\subsection{Tree Packing}\label{sec:tree-packing}
\begin{theorem} [Implicit in prior work~\cite{daga2019distributed, ghaffari2016mst, thorup2007fully}]
  \label{thm:treePacking}
  Let $G$ be any weighted $n$-node graph $G$ where the minimum cut has value $\lambda$ and let $\kappa > 0$ be a fixed constant. There is a randomized $\poly(\log n)$-round Minor-Aggregation algorithm that computes and distributedly stores a collection of spanning trees $T_1, T_2, \ldots, T_{\Theta(\log n)}$, such that with probability at least $1 - 1 / n^{\kappa}$ we have: for each cut $C$ of $G$ that has value at most $1.05\lambda$, the cut $C$ 2-respects at least one tree $T_i$ from the collection. That is, there exists $i$ such that at most $2$ edges of $T_i$ are in the cut $C$.
\end{theorem}
\begin{proof}[Proof Sketch]
  We note that this result is implicit in prior work, e.g., in the work of Daga et al.~\cite{daga2019distributed} by replacing the $\tl{O}(D+\sqrt{n})$-round $(1+\eps)$-minimum-cut approximation and minimum spanning tree algorithms with the immediate  $\tl{O}(1)$-round Minor-Aggregation algorithms explained in Ghaffari and Haeupler~\cite{ghaffari2016mst}. For completeness, we provide a brief proof sketch, without delving into the smaller details, which are explained in \cite{daga2019distributed}. 

  We treat any weighted graph as an unweighted graph with multiplicities, by replacing each edge $e$ with integer weight $w(e)$ with $w(e)$ parallel edges. We first compute a $1.01$ approximation $\bar{\lambda}$ of the min-cut value $\lambda$ using the randomized algorithm in $\poly(\log n)$ Minor-Aggregation rounds. We have two cases: 

  \begin{itemize}
  \item[(A)] If $\bar{\lambda} = O(\log n)$, then we perform an $I$-iteration greedy minimum spanning tree packing, where $I=2\bar{\lambda} \log m = O(\log^2 n)$.  That is, for each iteration $i\in [1, I]$, we set the cost $c(e)$ of each edge $e$ equal to the number of trees $T_1$, \dots, $T_{i-1}$ that contain edge $e$, and we compute a minimum-cost spanning tree using the MST algorithm $\poly(\log n)$ Minor-Aggregation rounds. This tree is recorded as $T_i$ in our collection and we proceed to the next iteration. Thorup~\cite{thorup2007fully} shows that this tree packing satisfies the desired property (in fact for an even wider range of cuts): namely, for every cut $C$ of $G$ that has value at most $1.1\lambda$, where $\lambda$ denotes the minimum cut value in $G$, at least one tree $T_i$ in the collection $2$-respects cut $C$.
  \item[(B)] Suppose that $\bar{\lambda} = \Omega(\log n)$. In this case, applying the greedy tree packing approach of (A) directly would require $O(\lambda \log n)$ iterations and would thus increase the round complexity to $\tl{O}(\lambda)$ rounds. To circumvent this, we can use a standard random sampling idea of Karger. Let $p = C\log n/\bar{\lambda}$, where $C$ is sufficiently large constant. Sample each edge with probability $p$, and let $H$ be the spanning subgraph that includes only sampled edges. By Karger's result~\cite{karger1994random}, it is known that, with high probability, $H$ has the min-cut value of at least $(1\pm 0.01)\lambda p$, and moreover, every cut $C$ of $G$ with value at most $1.05$, the value of the same cut in $H$ is at most $(1.05+0.01) \lambda p$. Hence, any $1.05$-minimum cut of $G$ remains a $1.1$-minimum cut of $H$. Furthermore, $H$ now has minimum cut value $O(\lambda p) = O(\log n)$ which makes it amenable to the greedy tree packing approach of (A), while keeping the $\poly(\log n)$ round complexity. We perform the $I$-iteration greedy tree packing of (A) on $H$, in $\poly(\log n)$ Minor-Aggregation rounds. Every cut $C$ of $G$ that is $1.05$-minimum cut of $G$ remains a $1.1$-minimum cut of $H$. Hence, by the result of Thorup~\cite{thorup2007fully}, at least one tree $T_i$ in the collection $2$-respects cut $C$. \qedhere
\end{itemize}
\end{proof}

\section{Extending the Minor-Aggregation Model}\label{sec:extending-minor-aggregation}

In this section, we extend the Minor-Aggregation model in two ways. First, we show how to do computations when virtual nodes and edges are added to the topology. Second, we show how to derandomize ubiquitous primitives like heavy-light decompositions and subtree sums.

\subsection{Virtual nodes}\label{sec:virtual-nodes}

Virtual graphs allow us to create a logical network $G\virt$ which can be implemented by only performing communication/computations on some underlying graph $G \subseteq G\virt$. For example, we might want to add a virtual node $v\virt$ (and its neighboring edges) to a graph $G$ and simulate computation in $G + v\virt$ without actually having $v\virt$ in the underlying network graph itself. Naturally, this simulation will inherently have some overhead; in this section, we show that the (multiplicative) overhead of adding $\beta$ virtual nodes and arbitrarily connecting them (to the rest of the graph or among themselves) is only $O(\beta + 1)$.

%Given a graph $G$, we can construct a \emph{virtual graph} by repeatedly adding a \emph{virtual node} and connecting it to an arbitrary set of neighbors. This is formalized in the following definition.

\begin{definition}\label{def:virtual-graph-extension}
  A virtual graph $G\virt$ extending $G = (V, E)$ is a graph whose node set can be partitioned into $V$ and a set of so-called \emph{virtual nodes} $V\virt$, i.e., $V(G\virt) = V \sqcup V\virt$. We say $G\virt$ has at most $\beta$ virtual nodes (as an extension of $G$) if $|V\virt| \le \beta$. Furthermore, each edge of $E(\virt)$ adjacent to at least one virtual node is called a \emph{virtual edge}.
\end{definition}

\noindent\textbf{Distributed storage of virtual graphs.} A virtual graph $G\virt$ (extending $G$) is distributedly stored in $G$ in the following way. All nodes are required to know the list of all (IDs of) virtual nodes. A virtual edge connecting a non-virtual $u$ and a virtual $v$ is only stored in $u$ (other nodes do not need to know about its existence). A virtual edge between two virtual nodes is required to be known by all nodes. 

% \begin{definition}[Virtual node addition]\label{def:virtual-node-addition}
%   Adding a virtual node to a graph $G = (V, E)$ corresponds to changing $G$ into $G\virt = (V \cup \{\mathrm{virt}\}, E \cup \bigcup_{v \in \Gamma\virt} \{ \mathrm{virt}, v\})$, where $\mathrm{virt}$ is a new ``virtual'' node, and $\Gamma\virt \subseteq V$ are the set of neighbors of $\mathrm{virt}$. Edges adjacent to virtual nodes are called ``virtual edges''.
% \end{definition}
\medskip
\noindent\textbf{Simulations on virtual graphs.} In a nutshell, we can add $\beta$ many virtual (arbitrarily interconnected) nodes to any graph $G$ and still simulate any Minor-Aggregation algorithm on the virtual graph with a $O(\beta + 1)$ blowup in the number of rounds.

\begin{theorem}\label{thm:simulating-virtual-nodes}
  Suppose $A\virt$ is a (deterministic) $\tau$-round Minor-Aggregation algorithm on a virtual graph $G\virt$ extending $G$, where $G$ is connected and the extension has at most $\beta$ virtual nodes. Any such $A\virt$ in $G\virt$ can be simulated with a (deterministic) $\tau \cdot O(\beta + 1)$-round Minor-Aggregation algorithm in $G$. Upon termination, each non-virtual node $v \in V(G)$ learns all information learned by $v$ and all virtual nodes.
\end{theorem}
\begin{proof}
  We show how to simulate a single Minor-Aggregation round on $G\virt$ via $O(\beta + 1)$ rounds of Minor-Aggregation on $G$. Suppose that $F\virt \subseteq E(G\virt)$ are the edges that are set to be contracted in the current round.

  First, we contract non-virtual edges $F\real := F\virt \cap E(G)$. Each supernode $s$ of $G / F\real$ learns in $\beta$ rounds the set of virtual nodes it is directly connected to via contracted edge $F\virt$. In slightly more detail, suppose we fix a virtual node $v\virt$. We contract $F\real$ and use a consensus OR-operator step where a node outputs $1$ if it is connected to $v\virt$ and $0$ otherwise. After the consensus step, each ($G$-node in each) $(G/F\real)$-supernode learns whether it is directly connected to $v\virt$ via $F\virt$ or not. This is repeated for all $\beta$ virtual nodes.

  Since each (node in each) supernode $s$ of $G / F\real$ knows its supernode ID, which virtual nodes $s$ is connected to, and the set of edges interconnecting the virtual nodes, it can compute its supernode ID in $G\virt / F\virt$ for instance, as the minimum ID of a connected virtual node (or its own ID if not connected to any virtual nodes). 

  We can now perform the consensus step in $G\virt$ within $\beta + 1$ rounds. In the first round, all supernodes that do not contain any virtual node perform their consensus step. This can be done in a single round in $G / F\real$. Next, we iterate over all virtual nodes $v\virt$. We contract the entire graph into a single node, and all nodes $v$ whose $(G\virt / F\virt)$-ID is $v\virt$ output their output $x_v$ as stipulated by $A\virt$ (other nodes output an identity element $\bot$). Clearly, after this step, the $(G\virt / F\virt)$-supernode containing $v\virt$ computes its consensus-step output $y$. Note that all nodes of $G$ learn everything that $v\virt$ learns. We repeat this for all $\beta$ and complete the consensus step.

  We now similarly perform the aggregation step in $G\virt / F\virt$. First, we specify in more detail who exactly simulates each edge. Consider an edge $e \in E(G\virt)$ with endpoints $a, b$ (which are nodes in $G\virt$). If both endpoints are non-virtual, the same edge exists in $G$ and simulates itself---it can learn its inputs $y_a, y_b$ directly. On the other hand, if $a$ is non-virtual and $b$ is virtual, then $a$ simulates the edge and knows both $y_a$ and $y_b$ (since all nodes know virtual nodes' $y$-values). Finally, if both endpoints are virtual, all nodes simulate those edges (which is valid since they all know $y_a, y_b$). This allows all edges $e$ to compute its outputs $z_{e, a}, z_{e, b}$ (the nodes/edges simulating them can compute them). Finally, we need to compute the aggregates of $z$-values in a similar way to the consensus step. In the first round, we compute the $z$-aggregates of $(G\virt / F\virt)$-supernodes that do not contain any virtual node. Then, in the next $\beta$ rounds, we process each virtual node $v\virt$ one by one, contract the entire graph, and compute the $z$-aggregate of the supernode $(G\virt / F\virt)$ with ID equal to $v\virt$.
\end{proof}

We now show a useful lemma stating that we can always replace a node with its virtual substitute, which can be useful since we can arbitrarily interconnect them to other (even non-virtual) nodes in $G$.
\begin{lemma}\label{lemma:virtual-node-replacement}
  Let $v \in V(G)$ be a node in $G$. In $O(1)$ deterministic Minor-Aggregation rounds, we can distributedly store a graph $G\virt$ where the node $v$ is replaced with a virtual node $v\virt$ such that $G\virt$ is a virtual graph extending $G$ with a single virtual node. Specifically, $v$ in $G$ and $v\virt$ in $G\virt$ have the same set of neighbors. If multiple edges connected $v$ with some neighbor, $G\virt$ will contain a single edge with a weight equal to the sum of such edges in $G$. 
\end{lemma}
\begin{proof}
  In a single round, by contracting all edges, we broadcast the ID of $v$ to all nodes. We can re-use this ID as the ID of the new virtual node (since $v$ and its incident deactivates after the replacement). In another round, without any contractions, each edge that is incident to $v$ reports this fact to its other endpoint (say) $w \neq v$ along with the edge weight. The node $w$ sums up the edge weights in its aggregation step. Using this, each node incident to $v$ knows it is also incident to the new virtual node $v\virt$, making the new graph distributedly stored (remember that $v$ or $v\virt$ does not need to know its incident edges, but its neighbors must know they are incident to the new virtual node).
\end{proof}

\subsection{Deterministic primitives and simulation}\label{sec:det-ops-and-simulation}

In this section, we develop several useful deterministic primitives. The proofs are fairly involved as they need to argue about low-level model-specific details and are deferred to \Cref{sec:tree-primitives}.

\begin{restatable}[Deterministic primitives]{lemma}{lemmaDetOps}\label{lemma:Det-Ops}
  Let $T$ be a tree and let $r \in V(T)$. Suppose each node $v$ has an $\tl{O}(1)$-bit private input $x_v$. There is a deterministic $\tl{O}(1)$-round Minor-Aggregation algorithm computing the following for each node $v$:
  \begin{itemize}\setlength\itemsep{0em}
  \item Heavy-light decomposition: $v$ learns its HL-info of the heavy-light decomposition rooted at $r$.
  \item Ancestor sum: $v$ learns $p_v := \bigoplus_{w \in \anc(v)} x_w$ where $\anc(v)$ is the set of ancestors of $v$ w.r.t. root $r$.
  \item Subtree sum: $v$ learns $s_v := \bigoplus_{w \in \desc(v)} x_w$ where $\desc(v)$ is the set of descendants of $v$ w.r.t. root $r$.
  \end{itemize}
\end{restatable}

We now state how to simulate (deterministic) Minor-Aggregation algorithms in (deterministic) CONGEST. The proof is deferred to \Cref{sec:tree-primitives}.
\begin{restatable}{theorem}{thmCongestSimulation}\label{thm:congest-simulation}
  Suppose $A$ is any deterministic $\tau$-round Minor-Aggregation algorithm and suppose $G$ is an $n$-node graph with diameter-$D$. We can simulate $A$ with a CONGEST algorithm on $G$ with the following guarantees:
  \begin{itemize}\setlength\itemsep{0em}
  \item Unconditionally, the simulation requires $\tau \cdot \tl{O}(D + \sqrt{n})$ rounds and is deterministic.~\cite{ghaffari2016distributed}
  \item Unconditionally, the simulation requires randomized $\tau \cdot \poly(\SQ(G)) \cdot n^{o(1)}$ rounds.~\citep{haeupler2022oblRouting}
  \item When $G$ is an excluded-minor graph (e.g., planar graph), the simulation requires $\tau \cdot \tl{O}(D)$ rounds and is deterministic.~\cite{haeupler2016low,ghaffari2021excluded}
  \item When the topology $G$ is known, the simulation requires randomized $\tau \cdot \tl{O}(\SQ(G))$ rounds.~\cite{haeupler2020shortcuts}
  \end{itemize}
\end{restatable}

\section{Warm-up: 1-Respecting Min-Cut}\label{sec:one-respecting-cut}

In this section, we show how to compute all 1-respecting cuts when given a spanning tree $T$ of a graph $G$ with a $\poly(\log n)$-round Minor-Aggregation algorithm on $G$.

\begin{theorem}\label{thm:1-respecting-cut-algorithm}
  Let $T = (V, E_T)$ be a rooted spanning tree of a weighted graph $G = (V, E_G)$. There exists a deterministic $\tl{O}(1)$-round Minor-Aggregation algorithm that computes $\Cut_{T, G}(e)$ for all $e \in E_G$ (each edge learns its cut value).
\end{theorem}
\begin{proof}
  Compute a heavy-light decomposition of $T$ (\Cref{lemma:Det-Ops}). We note that the cut value $\Cut_{T, G}(e)$ can be computed as the sum of contributions from all (graph) edges $e \in E_G$, where the contribution of each edge $e = \{u, v\} \in E_G$ is as follows. Due to $e$, we need to increase the cut values by $+w(e)$ for all edges on the unique $u$-to-$v$ path in $T$. Equivalently, the contribution of $e$ to a tree-edge $f := \{\parent(x), x\}$ can be calculated as the subtree sum (with the $+$-aggregator) of $A_e \in \R^V$ defined as follows. We define a node vector $A_e \in \R^V$ with all values $0$ except $A_e(u) := w(e), A_e(v) := w(e), A_e(\LCA(u, v)) := - 2 w(e)$ if the edge is not ancestor-descendant. Otherwise, if $e$ is ancestor-descendant with $u$ farther from the root, then $A_e(u) := w(e)$ and $A_e(v) := - w(e)$. It is immediate that, for each non-root node $x \in V$, the contribution of $e$ to $\Cut_{T, G}(\{parent(x), x\})$ (i.e., to $x$'s parent edge) corresponds to the sum of values in $A_e$ over all $x$'s descendants. Therefore, by defining $A := \sum_{e \in E_G} A_e \in \R^V$ and assuming it can be computed, we conclude that $\Cut_{T, E}(\cdot)$ can be computed as the subtree sum over the vector $A$. Since the subtree sum can be deterministically computed in $\tl{O}(1)$ Minor-Aggregation rounds (\Cref{lemma:Det-Ops}), we reduced our problem to computing and distributedly storing $A \in \R^V$.

  We now discuss how to compute $A$. Initially, we set all $A(x) \gets 0$ for all $x \in V$. First, for each graph edge $e = \{u, v\} \in E_G$ we increase $A(u)$ and $A(v)$ by $+w(e)$. In other words, so far, $A(x)$ is equal to the sum of weights of all graph edges incident to $x$. This can be achieved with a single Minor-Aggregation round (without any contractions).

  Second, for each graph edge $e = \{u, v\} \in E_G$ let $l = \LCA(u, v)$ be the lowest common ancestor of the endpoints. Our goal is to decrease $A(l)$ by $-w(e)$, and perform this operation over all $e \in E_G$. This is achieved as follows. Fix some edge $\{u, v\} \in E_G$. If $e$ is an ancestor-descendant edge (i.e., $\LCA(u,v) \in \{u, v\}$), then we handle this case locally. Otherwise, we have that $l := \LCA(u, v)$ is in the HL-info list of (at least) one of $\{u, v\}$ (\Cref{lemma:lca-labels}); suppose without loss of generality this is $u$. Then, we say that ``$u$ is \textbf{responsible} for updating the \textbf{target} $l$ by a \textbf{delta} of $-w(e)$''. We use a subtree-sum operation (\Cref{lemma:Det-Ops}) to find, for each node $u$, an associative array that maps each ancestor node $v \in \anc(u)$ to the total sum of deltas where the responsible node is a descendant of $u$ and the target is $v$. Initially, each node $v$ initializes its private input (for the subtree-sum operation) with the list of targets it is responsible for updating. Then, the result of the subtree sum for a node $v$ is the sum of the private inputs over all $\desc(v)$, with the resulting associative array restricted to the domain of ancestors $\anc(v)$. Immediately, the resulting associative array is supported on the endpoints of light edges of the root-to-$u$ path (i.e., all other values are $0$). Therefore, it is supported on a $\tl{O}(1)$-sized set (\Cref{lemma:hl-has-logn-light-edges}). Moreover, the same holds for any partial result that pops up during the computation: if we are aggregating the arrays of two nodes $u$ and $v$, we can ignore all entries that are not on light edges of the root-to-$\LCA(u,v)$ path (i.e., the intersection of light edges on root-to-$u$ and root-to-$v$ path). This means that the operation always fits within $\tl{O}(1)$-bits, meaning it can be implemented via an aggregation operation. This concludes the computation of the vector $A$, which can be achieved in $\tl{O}(1)$ Minor-Aggregation rounds.

\end{proof}

\section{Path-to-path 2-Respecting Min-Cut}\label{sec:path-to-path-cut}

This section shows how to compute the minimum 2-respecting cut between two paths $P$ and $Q$ (adjoined with a root for orientation purposes, see \Cref{fig:path-to-path}). We formalize this notion in the following result, which is the main result of this section.
\begin{figure}
  \centering
  \includegraphics[width=0.6\textwidth]{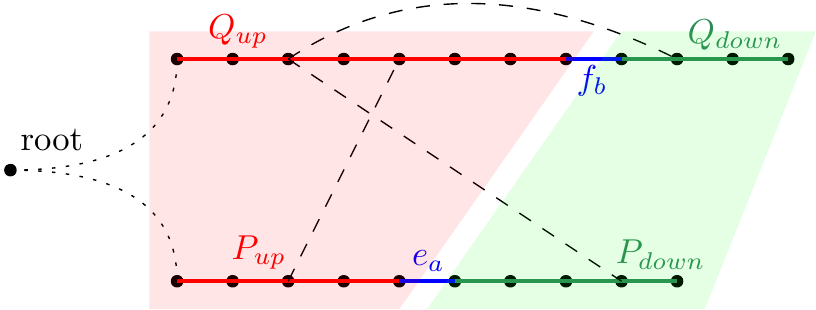}
  \caption{A path-to-path instance. The dashed edges between $P$ and $Q$ are cross-path, while the rest of the edges are same-path. The midpoint of $P$ is $e_a$ and its best reponse is $f_b$. The recursive calls are issued on $P_{\text{up}} \times Q_{\text{up}}$ (red area) and $P_{\text{down}} \times Q_{\text{down}}$ (green area).}
  \label{fig:path-to-path}
\end{figure}

\begin{theorem}\label{thm:path-to-path-algorithm}
  Suppose $G$ is a weighted graph and $T \subseteq G$ is $G$'s (rooted) spanning tree. Moreover, $T$ is composed of a root $r$, and two descending paths $P, Q$. There exists a deterministic $\tl{O}(1)$-round Minor-Aggregation algorithm on $G$ that computes the minimum of 1-respecting $\min_{e \in E(P) \cup E(Q)} \Cut_{T, G}(e)$ and 2-respecting cuts $\min_{e \in E(P), f \in E(Q)} \Cut_{T, G}(e, f)$. 
\end{theorem}

We number the edges of $P$ as $e_1, e_2, \ldots, e_{|P|}$ in order of increasing depth (with $|P|$ denoting the length of the path). Similarly, we let $f_1, \ldots, f_{|Q|}$ be the edges of $Q$.

The main idea is to import the state-of-the-art centralized techniques into the distributed setting with the help of the Minor-Aggregation model. Specifically, we first fix $e_a$ to be the midpoint edge of $P$ (i.e., $a := \lfloor |P| / 2 \rfloor)$ and let $f_b$ be the \textbf{best response} to $f_a$, meaning the edge $f_b \in E(Q)$ that minimizes $\Cut(e_a, f_b)$. Then, either $(e_a, f_b)$ is the pair that minimizes the 2-respecting cut, or the minimizing pair can be found on either $\{e_1, \ldots, e_{a-1}\} \times \{f_1, \ldots, f_{b-1}\}$ or $\{e_{a+1}, \ldots, e_{|P|}\} \times \{f_{b+1}, \ldots, f_{|Q|}\}$ (i.e., the minimizing pair is either entirely on one side or on the other side of $(e_a, f_b)$). This last property easily follows from the so-called \emph{Monge property} and was recently observed by \cite{mukhopadhyay2020weighted}.
\begin{fact}[Claim 7.1. in the full version of \cite{dory2021distributed}; \cite{gawrychowski2021note}]\label{fact:monge-property}
  For all $i \le i', j \le j'$ we have:
  \begin{align*}\Cut_{T, G}(e_i, f_j) + \Cut_{T, G}(e_{i'}, f_{j'}) \le \Cut_{T, G}(e_{i'}, f_{j}) + \Cut_{T, G}(e_i, f_{j'}) .\end{align*}
\end{fact}

\noindent\textbf{Notation.} We introduce some (section-specific) notations. An edge $\{u, v\} \in E_G$ is a \textbf{cross-path} if it has one endpoint on both $V(P)$ and $V(Q)$; otherwise, it is a \textbf{same-path} edge. Given a graph $G$ and a set of nodes $D \subseteq V(G)$, we denote by $G - D$ the subgraph with all nodes in $D$ removed (their incident edges are also removed). Furthermore, given a path-to-path instance $T \subseteq G$, we say the instance is \textbf{separable} if $G - \{ r, \mytop(P), \bottom(P), \mytop(Q), \bottom(Q) \}$ has no cross-path edges, where $\mytop(\cdot), \bottom(\cdot)$ of a rooted path $X$ represent the closest- and furthest-away nodes on $X$ from the root. It is important to observe that the instance is not separable if and only if $G - \{ r, \mytop(P), \bottom(P), \mytop(Q), \bottom(Q) \}$ is connected (or the paths have less than 3 nodes).

To calculate the best response of an edge (i.e., given $e \in E(P)$, calculate $\min_{f \in E(Q)} \Cut(e, f)$), we use the following result.
%The following result enables us to compute 2-respective covers when one edge is fixed.
\begin{lemma}\label{lemma:one-sided-cov-algorithm}
  Assume the setting of \Cref{thm:path-to-path-algorithm} and let $e_{\text{fix}} \in E(P)$ be a fixed edge. There is a deterministic algorithm where each edge $f \in E(Q)$ learns $\Cov(e_{\text{fix}}, f)$ that runs in $\tl{O}(1)$-round Minor-Aggregation algorithm on $G$.
\end{lemma}
\begin{proof}
  First, we compute depths for each node on $P, Q$ using a single subtree-sum operation (initialize all private values to $1$ and use \Cref{lemma:Det-Ops}'s subtree sum with the $+$-aggregation on $T$). Then, for each cross-path edge $\{u, v\} = e \in E(G)$ with $u \in V(P), v \in V(Q)$, we perform the following. If $u$ is below $\bottom(e_{\text{fix}})$ (specifically, $\depth(u) \ge \depth(\bottom(e_u))$), we add $+w(e)$ to the label of $v$. Note that this operation can be performed in a single Minor-Aggregation round. Finally, the subtree sum of labels at a node $v$ represents the $\Cov_{T, G}(e_{\text{fix}}, \{v, \parent(v)\} )$. Therefore, we compute the subtree sum (\Cref{lemma:Det-Ops}) and obtain the required result in $\tl{O}(1)$ Minor-Aggregation rounds.
\end{proof}

Next, we develop an algorithm that solves \emph{separable} instances without any recursive calls.
\begin{lemma}\label{lemma:hollow-2-resp-cut}
  Assume the setting of \Cref{thm:path-to-path-algorithm} and suppose that $G - \{ r, \mytop(P), \bottom(P), \mytop(Q), \bottom(Q) \}$ has no cross-path edges. There exists a deterministic $\tl{O}(1)$-round Minor-Aggregation algorithm on $G$ that computes the minimum 2-respecting cut $\min_{e \in E(P), f \in E(Q)} \Cut_{T, G}(e, f)$.
\end{lemma}
\begin{proof}
  We show that, since the instance is separable, $\Cut_{T, G}(e, f)$ is \emph{separable} in the following sense: there exist two functions $F_P : E(P) \to \R$ and $F_Q : E(Q) \to \R$ such that $\Cut(e, f) = F_P(e) + F_Q(f)$ for all $e, f$. Due to $\Cut(e, f) = \Cov(e) + \Cov(f) - 2 \Cov(e, f)$ (\Cref{lemma:2-cut-vs-cov}) and $\Cov(e), \Cov(f)$ being trivially separable, it is sufficient to prove that $\Cov(e, f)$ is separable. We argue this by showing the contribution to $\Cov(\cdot, \cdot)$ from each type of allowable edges is separable.

  First, we note that any edge originating from $\mytop(P)$, $\mytop(Q)$ or $r$ does not contribute to $\Cov(e, f)$, making the contribution of such edges trivially separable. Second, the same-path edges do not contribute to $\Cov(e, f)$, making them trivially separable. Finally, there might exist edges that are incident to $\bottom(P)$ or $\bottom(Q)$. However, this is also separable: consider an edge $c := \{ \bottom(P), x \in V(Q) \}$; the contribution of $c$ to $\Cov(e, f)$ is $w(c)$ if $f$ is deeper than $x$ and $0$ otherwise; making it separable. The $\{ \bottom(Q), x \in V(P) \}$ case is symmetric. This covers all allowable types of edges.

  Finally, the functions $F_P, F_Q$ are easily computable in $\tl{O}(1)$ Minor-Aggregation rounds. First, we compute the $\Cov(e)$ and $\Cov(f)$ using the 1-respecting min-cut algorithm (\Cref{thm:1-respecting-cut-algorithm}). Following the case analysis from above, it is easy to calculate the contributions of cross-edges adjacent to $\bottom(P)$ or $\bottom(Q)$. Therefore, after we computed (and distributedly stored as edge vectors) $F_P$ and $F_Q$, we minimize each side separately and broadcast the result to $G$. This is the minimizing 2-respecting cut since $\min_{e \in E(P), f \in E(Q)} \Cut_{T, G}(e, f) = \min_{e \in E(P), f \in E(Q)} F_P(e) + F_Q(f) = \min_{e \in E(P)} F_P(e) + \min_{f \in E(Q)} F_Q(f)$.
\end{proof}

\begin{lemma}\label{lemma:path-to-path-algorithm}
  Assume the setting of \Cref{thm:path-to-path-algorithm} and suppose the instance is not separable. There exists a deterministic $\tl{O}(1)$-round Minor-Aggregation algorithm on $G - \{ r, \mytop(P), \bottom(P), \mytop(Q), \bottom(Q) \}$ that computes the minimum of 1-respecting and 2-respecting cuts $\min_{e \in E(P), f \in E(Q)} \Cut_{T, G}(e, f)$.
\end{lemma}

\begin{proof}
  \textbf{Algorithm.} We now present the algorithm facilitating this result and then prove its runtime and correctness.

  \begin{enumerate}  
  \item Note that computing 1-respecting cuts, i.e., values $\Cut_{T, G}(e)$ for each $e \in E(P) \cup E(Q)$ can be computed in $\tl{O}(1)$ Minor-Aggregation rounds via \Cref{thm:1-respecting-cut-algorithm}.

  \item Then, we note that if $|P| \le 10$ or $|Q| \le 10$, we can solve the problem in $\tl{O}(1)$ rounds: iterate over each edge of the smaller path and using the 2-respective fixed-edge algorithm (\Cref{lemma:one-sided-cov-algorithm}) to find all possible cover values. A final min-aggregation is required to compute the result.
  
  %, can be done as follows. First, we compute depths for each node on $P, Q$ (using a subtree-sum operation \alert{CITE}). For each edge $e = \{u, v\} \in E(G)$ with both endpoints on $P$ or $Q$, we add $+w(e)$ to the deeper endpoint (farther from the root) and add $-w(e)$ to endpoint closer to the root. For each edge going between $P$ and $Q$, we add $+w(e)$ to both endpoints. Note that these values can be computed in a single Minor-Aggregation round since endpoints know their depth. Computing the subtree sum of these weights corresponds exactly with the 1-respecting cut of its parent edge.
  
  \item \textbf{Midpoint $e_a$ and its best response $f_b$.} Let $a := \lfloor \frac{|P|}{2} \rfloor$, making $e_a$ the midpoint edge of $P$. Then, for each $f \in E(Q)$, we compute $\Cov(e_a, f_j)$ using \Cref{lemma:one-sided-cov-algorithm}. Finally, we let $f_{b}$ be the edge that minimizes $b := \arg\min_j Cut(e_a, f_j) = Cov(e_a) + Cov(f_j) - 2 Cov(e_a, f_j)$, and compute $\Cut(e_a, f_b)$. Using $O(1)$ Minor-Aggregation rounds, all nodes and edges in $P$ and $Q$ learn $f_b$ and $\Cut(e_a, e_b)$. Furthermore, let $e_{a'}$ be the best response to $f_b$, i.e., the edge that minimizes $a' := \arg\min_i \Cut(e_i, f_b) = Cov(e_i) + Cov(f_b) - 2 Cov(e_i, f_b)$. All nodes and edges on $V(P) \cup V(Q)$ learn $\Cut(e_{a'}, f_b)$. % I guess we need this to completely exclude e_a, f_b from further computations

  \item Split $P$ into two paths $P_{\text{up}} = (e_1, \ldots, e_{a-1})$ and $P_{\text{down}} = (e_{a+1}, \ldots, e_{|P|})$ (excluding $e_a$). Similarly, split $Q$ into $Q_{\text{up}} = (f_1, \dots, f_{b-1})$ and $Q_{\text{down}} = (f_{b+1}, \ldots, b_{|Q|})$. Let $p_{-1}, q_{-1}$ be the nodes of $P_{\text{up}}, Q_{\text{up}}$ farthest away from the root, resp. We replace $p_{-1}$ and $q_{-1}$ with a virtual node (\Cref{lemma:virtual-node-replacement}). This allows us to add arbitrary edges incident to them.

  \item \textbf{Constructing cut-equivalent $G_{\text{up}}$.} We now construct (and distributedly store) a graph $G_{\text{up}}$ which preserves 1- and 2-respecting cover values (and therefore, also cut values) for edge pairs $E(P_{\text{up}}) \times E(Q_{\text{up}})$. First, let $W$ be the total weight of all edges between (any node of) $P_{\text{down}}$ and (any node of) $Q_{\text{down}}$. We insert an edge between $p_{-1}$ and $q_{-1}$ of weight $W$. Second, for each node $v \in V(Q_{\text{up}})$ let $W_v$ be the total weight of edges between $v$ and (any node of) $P_{\text{down}}$. We insert an edge between $v$ and $p_{-1}$ of weight $W_v$. Finally, for each node $u \in V(P_{\text{up}})$ let $W_u$ be the total weight of edges between $u$ and (any node of $Q_{\text{down}})$. We insert an edge between $u$ and $q_{-1}$ of weight $W_u$. Let $T_{\text{up}} := T[V(P_{\text{up}}) \cup V(Q_{\text{up}}) \cup \{root\}]$ be the restriction of $T$ to $G_{\text{up}}$ (i.e., with all inclusive descendants of $e_a, f_b$ contracted). The following is immediate by construction. %
  %The reason for this change is because, if we denote with $G_{\text{up}}$ all edges between $P_{\text{up}}$ and $Q_{\text{up}}$, that now $(P_{\text{up}}, Q_{\text{up}}, G_{\text{up}})$ preserves all 1- and 2-respecting cuts with respect to $(P, Q, G)$. In other words, for all $e \in E(P_{\text{up}}), f \in E(Q_{\text{up}})$ we have that $\Cut_{T, G_{\text{up}}}(e, f) = \Cut_{T, G}(e, f)$ and similarly for 1-respecting cuts. \alert{maybe argue why this is the case}
  \begin{fact}\label{fact:up-cut-equivalent}
    For all pairs of edges $e \in E(P_{\text{up}})$ and $f \in E(Q_{\text{up}})$ we have that $\Cov_{T_{\text{up}}, G_{\text{up}}}(e, f) = \Cov_{T, G}(e, f)$ and $\Cut_{T_{\text{up}}, G_{\text{up}}}(e, f) = \Cut_{T, G}(e, f)$.
  \end{fact}

  \item \textbf{Constructing cut-equivalent $G_{\text{down}}$.} Similarly, we construct (and distributedly store) $G_{\text{down}}$, which preserves 1- and 2-respecting cover values (and therefore, also cut values) for edge pairs $E(P_{\text{down}}) \times E(Q_{\text{down}})$. We define $G_{\text{down}}$ as the induced graph $G[ V(P_{\text{down}}) \cup V(Q_{\text{down}}) ]$ (i.e., $G$ restricted to edges going between $P_{\text{down}}$ and $Q_{\text{down}}$). We also add a virtual root node $r_{\text{down}}$ and connect it with arbitrary weight (since it's not considered) to the top nodes of $P_{\text{down}}$ and $Q_{\text{down}}$. It is easy to see (easier than for $G_{\text{up}}$) that $G_{\text{down}}$ preserves 1- and 2-respecting cover values and cuts. Defining $T_{\text{down}} := T[V(P_{\text{down}}) \cup V(Q_{\text{down}}) \cup \{r_{\text{down}}\}]$, this is formalized as follows. %  
  \begin{fact}\label{fact:down-cut-equivalent}
    For all pairs of edges $e \in E(P_{\text{down}})$ and $f \in E(Q_{\text{down}})$ we have that $\Cov_{T_{\text{down}}, G_{\text{down}}}(e, f) = \Cov_{T, G}(e, f)$ and $\Cut_{T_{\text{down}}, G_{\text{down}}}(e, f) = \Cut_{T, G}(e, f)$.
  \end{fact}
  
  \item \textbf{Recursion on $G_{\text{up}}$.} We now recursively compute the minimum 2-respecting cut $\min_{e \in E(P_{\text{up}}), f \in E(Q_{\text{up}})}\allowbreak \Cut_{T, G}(e, f)$. First, if $T_{\text{up}} \subseteq G_{\text{up}}$ is a separable instance (which can be checked in a single Minor-Aggregation round), we solve the problem without recursion via (\Cref{lemma:hollow-2-resp-cut}). Otherwise, we recursively call (the same \Cref{lemma:path-to-path-algorithm}) on $T_{\text{up}} \subseteq G_{\text{up}}$ to recover the result with an algorithm that operates on $G_{\text{up}} - \{ r, \mytop(P_{\text{up}}), \bottom(P_{\text{up}}), \mytop(Q_{\text{up}}), \bottom(Q_{\text{up}}) \}$. Furthermore, all virtual nodes (introduced in the current recursive call) are contained within the deleted nodes, hence the same algorithm also runs on $G$ without the need to eliminate any virtual nodes. This prevents simulation cascade.

  %that there exists an edge in $G$ (an edge in $G_{\text{up}}$ is not sufficient) between $P_{\text{up}}$ and $Q_{\text{up}}$. In this case we recursively call \Cref{thm:path-to-path-algorithm} on $P_{\text{up}}, Q_{\text{up}}, G_{\text{up}}$, which runs an $\tl{O}(1)$-round algorithm on $G_{\text{up}}$. Since there can be at most $4$ virtual nodes in $G$ (top and bottom of $P$ and $Q$), this can be compiled into a $\tl{O}(1)$-round algorithm on the non-virtual nodes of $G_{\text{up}}$ (\Cref{thm:simulating-virtual-nodes}). Now, suppose that (2) there are no edges between $P_{\text{up}}$ and $Q_{\text{up}}$ in $G$. This immediately implies that $\Cov_{T, G}(e, f)$ is constant as $e, f$ range over $E(P_{\text{up}}), E(Q_{\text{up}})$. Noting that $\Cut(e, f) = \Cov(e) + \Cov(f) - 2 \Cov(e, f)$, we $O(1)$ Minor-Aggregation rounds to compute $\min_{e \in E(P_{\text{up}})} \Cov(e)$, $\min_{f \in E(Q_{\text{up}}) }\Cov(f)$ and this constant (single value) $\Cov_{\text{const}} := \Cov(e, f)$ for any $e \in E(P_{\text{up}}), f \in E(Q_{\text{up}})$, allowing us to find the minimizing value of the 2-respecting cut $\min_{e \in E(P_{\text{up}}), f \in E(Q_{\text{up}})} \Cut(e, f) = \min_{e \in E(P_{\text{up}})} \Cov(e) + \min_{f \in E(Q_{\text{up}})} \Cov(f) + \Cov_{\text{const}}$. Considering both cases (1) and (2), we conclude with the minimum 2-respecting cut between $P_{\text{up}}$ and $Q_{\text{up}}$ requiring $\tl{O}(1)$ Minor-Aggregation rounds.

  \item \textbf{Recursion on $G_{\text{down}}$.} We, analogously, compute the minimum 2-respecting cut on $T_{\text{down}} \subseteq G_{\text{down}}$. If $G_{\text{down}} - \{ r_{\text{down}}, \mytop(P_{\text{down}}), \bottom(P_{\text{down}}), \mytop(Q_{\text{down}}), \bottom(Q_{\text{down}}) \}$ is not connected (i.e., the instance is separable) we avoid recursion and use \Cref{lemma:hollow-2-resp-cut} to solve the instance. Otherwise, we use a recursive call that can be immediately run on $G$ without any translation.

  \item \textbf{Eliminating virtual nodes.} Finally, we remind that the final algorithm (by assumption) is required to avoid (not use) the nodes $D := \{ r, \mytop(P), \bottom(P), \mytop(Q), \bottom(Q) \}$ as they are potentially virtual. However, this is easy: the recursive calls already do not use any node in $D$ by assumption. We only need to eliminate the usage of $D$ in the non-recursive parts of the algorithm. But this is immediate from \Cref{thm:simulating-virtual-nodes} with a multiplicative blowup of $O(1)$ since we introduced only $\beta \le O(1)$ many virtual nodes (in the current recursive call) and $G - D$ is connected because $|P| > 10, |Q| > 10$ and the instance is not separable. This concludes the description of the algorithm.
  \end{enumerate}
  
  \noindent\textbf{Runtime analysis.} We point out that the algorithms computing the 2-respecting min-cuts on $G_{\text{up}}$ and $G_{\text{down}}$ are node-disjoint. Therefore, we can schedule them simultaneously. Suppose that, excluding the recursive calls, the maximum number of Minor-Aggregation rounds the other operations within the current call take is $\tl{O}(1) \le C \log_2^C n$, for some sufficiently large constant $C > 0$. Furthermore, since the length of the path $P$ at least halves in each subsequent recursion level, the depth of the recursion is $O(\log n)$. We show, by induction, that the algorithm terminates in at most $(\log_2 |P|) (C \log_2^C n)$ rounds (for some universal sufficiently-large constant $C > 0$). Clearly, the assumption is true on the leaves of the recursion as every step is $\tl{O}(1)$ rounds. Furthermore, for some paths $P, Q$, the two recursive calls (together) take at most $(\log_2 |P|/2) (C \log_2^C n)$ which, with the extra processing, gives us a bound of $(\log_2 |P|/2) (C \log_2^C n) + C \log_2^C n = (\log_2 |P|) (C \log_2^C n)$ rounds, as required.

  \smallskip

  \noindent\textbf{Correctness Analysis.} First, we note that the algorithm only checks some number of existing 2-respecting cuts, hence it can never report an answer that is smaller than the optimum solution. We only need to show that it successfully managed to find the optimum. Suppose that $(e^*, f^*) \in E(P) \times E(Q)$ is the pair that minimizes the 2-respecting cut. If $e^* = e_a$ (i.e., the optimal edge is the midpoint edge of $P$), then (by definition) the best response $f_b$ gives the optimal solution. Similarly, if $f^* = f_b$, then the best response $e_{a'}$ gives the optimal solution. Now, we show that an optimum solution $(e', f')$ must exist where both $e'$ and $f'$ must be either both closer (or equal) to the root of both farther away (or equal) from the root than $e_a, f_b$. This follows from the Monge property. If both $e^*$ and $f^*$ are closer (or both are farther away), then we are done. Now, assume $e^*$ is closer and $f^*$ is farther away, then by \Cref{fact:monge-property} we have $\Cut(e^*, f_b) + \Cut(e_a, f^*) \le \Cut(e^*, f^*) + \Cut(e_a, f_b)$, which implies $\Cut(e^*, f_b) \le \Cut(e^*, f^*) + [\Cut(e_a, f_b) - \Cut(e_a, f^*)]$. Note that the term inside $[\cdot]$ is non-positive since $f_b$ is the best response to $e_a$, hence $\Cut(e^*, f_b) \le \Cut(e^*, f^*)$, implying that $(e', f') = (e^*, f_b)$ is an optimum solution and satisfies the requirements. The case where $e^*$ is farther and $f^*$ is closer is analogous. Therefore, the optimum solution can be found in either $E(P_{\text{up}}) \times E(Q_{\text{up}})$ or $E(P_{\text{down}}) \times E(Q_{\text{down}})$. However, we assumed we solved the problems (recursively or via separable instances) on $G_{\text{up}}$ and $G_{\text{down}}$. Since they are cut-equivalent to the original graph (\Cref{fact:down-cut-equivalent} and \Cref{fact:up-cut-equivalent}, we conclude we found the optimum.
\end{proof}

Finally, with all the ingredients in place, we can directly argue \Cref{thm:path-to-path-algorithm}.
\begin{proof}[Proof of \Cref{thm:path-to-path-algorithm}]
  If the instance is separable, or $|P| \le 10$, or $|Q| \le 10$, we can trivially solve the problem using \Cref{lemma:one-sided-cov-algorithm} and \Cref{lemma:hollow-2-resp-cut}. Otherwise, we use \Cref{lemma:path-to-path-algorithm} and conclude.
\end{proof}

\section{Star 2-Respecting Min-Cut}\label{sec:star-cut}

In this section, we show how to compute the minimum 2-respecting cut between $k$ paths $P_1, P_2, \ldots, P_k$ (adjoined with a root for orientation purposes, see \Cref{fig:star}). We call such an input a ``star instance'' and formalize it in the following definition.

\begin{definition}
  A \textbf{star instance} $\{ P_i \}_{i=1}^k \subseteq T \subseteq G$ is composed of the following. Suppose $G$ is a weighted graph and $T \subseteq G$ is $G$'s (rooted) spanning tree. Moreover, $T$ is composed of exactly a root $r$, and $k$ (disjoint) descending paths $P_1, P_2, \ldots, P_k$.
\end{definition}

The following result formalizes the goal; it is the main result that will be proved later in the section once sufficient tooling is developed.
\begin{restatable}{theorem}{thmStarAlgorithm}\label{thm:star-algorithm}
  Given a star instance $\{ P_i \}_{i=1}^k \subseteq T \subseteq G$, there exists a deterministic $\tl{O}(1)$-round Minor-Aggregation algorithm on $G$ that computes the minimum of 1-respecting cuts and 2-respecting cuts $\min_{i < j} \allowbreak \min_{e \in E(P_i), f \in E(P_j)} \allowbreak \Cut_{T, G}(e, f)$.
\end{restatable}

\begin{figure}
  \centering
  \includegraphics[width=0.18\textwidth]{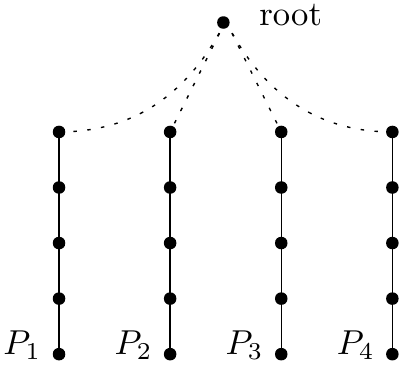}
  \caption{A star instance with $k = 4$ paths.}
  \label{fig:star}
\end{figure}

\subsection{A structural result: path interest}\label{sec:path-interest}

We now derive a significant structural result observed by Mukhopadhyay and Nanongkai~\cite{mukhopadhyay2020weighted}, which allows us to solve star instances efficiently: if the 2-respecting cut determined by the pair of edges $e, f$ has a smaller value than any 1-respecting cut, than more than half of the edges covering $e$ also cover $f$. The analogous claim also holds if we only consider only the cross edges (which will allow us to avoid certain technical issues later).

\medskip

\noindent\textbf{Notation.} Given a star instance $\{ P_i \}_{i=1}^k \subseteq T \subseteq G$, an edge $\{u, v\} = e \in E(G)$ is a \textbf{cross-edge} if its endpoints $u, v$ are in different paths $u \in V(P_i)$ and $v \in V(P_j)$ for $i \neq j$. For $e, f \in E(T)$ on different paths, we define $\CrossCov(e, f)$ as the sum of weights of all cross-edges $\{u, v\} \in E(G)$ such that the unique $u$-to-$v$ path in $T$ covers both $e$ and $f$. We also define $\CrossCov(e) := \CrossCov(e, e)$.

\begin{lemma}\label{lemma:best-cut-has-interest}
  Let $\{ P_i \}_{i=1}^k \subseteq T \subseteq G$ be a star instance. Given two path edges $e \in E(P_i), f \in E(P_j), i \neq j$, if $\Cut(e, f)$ is smaller than any 1-respecting cut, then $\CrossCov(e, f) > \CrossCov(e) / 2$.
\end{lemma}
\begin{proof}
  Since $\Cut_{T, G}(e, f)$ is smaller than any 1-respecting cut, we have that $\Cov(e, f) > Cov(e) / 2$. (\Cref{lemma:interesting-path-vs-cov}). Since $e$ and $f$ are on different paths, we have that $\Cov(e, f) = \CrossCov(e, f)$. Furthermore, since the set of cross-edges is a subset of $E(T)$, we have $\Cov(e) \ge CrossCov(e)$. Combining, we get $\CrossCov(e, f) = \Cov(e, f) > \Cov(e) / 2 \ge \CrossCov(e) / 2$.
\end{proof}

Previously, we only talked about pairs of edges. We generalize this notion to pairs of paths (called \emph{interested paths}) with the following definition.
\begin{definition}
  Let $\{ P_i \}_{i=1}^k \subseteq T \subseteq G$ be a star instance. Given two edges tree edges $e, f$, we say $e$ is \textbf{$\alpha$-interested} (for some $0 < \alpha < 1$) if $\CrossCov(e, f) > \alpha \cdot \CrossCov(e)$. Similarly, $P_i$ is $\alpha$-interested in $P_j$ if some edge $e \in E(P_i)$ is $\alpha$-interested in $P_j$. Furthermore, we call pairs of $1/2$-interested edges (or paths) \textbf{strongly interested} and $1/5$-interested edges (or paths) \textbf{weakly interested}.
\end{definition}

% We number the edges of $P$ as $e_1, e_2, \ldots, e_{|P|}$ in order of increasing depth (with $|P|$ denoting the length of the path). Similarly, we let $f_1, \ldots, f_{|Q|}$ be the edges of $Q$. The following is the main result of this section.

The salient reason why path interest helps in solving star instances is the fact that a path is only interested in a few other paths. This greatly reduces the number of pairs of paths we need to consider when searching for optimal 2-respecting cuts.
\begin{lemma}\label{lemma:path-weak-interest-cardinality}
  Let $\{ P_i \}_{i=1}^k \subseteq T \subseteq G$ be a star instance. Each path $P_i$ is weakly interested in at most $O(\log n)$ paths $\{ P_j \}$.
\end{lemma}
\begin{proof}
  We first show two subclaims and then proceed the prove the result.

  \paragraph{Subclaim 1.} Suppose that $e_1, e_2 \in E(P_i)$ where $e_1$ is closer to the root. Let $f \in E(P_j), i \neq j$ be an edge in a different path. If $e_2$ is \emph{not} $1/10$-interested in $f$, but $e_1$ is weakly interested in $f$, then $\CrossCov(e_1) \ge 1.1 \cdot \CrossCov(e_2)$.

  \paragraph{Proof of Subclaim 1.}  Let $p$ be the unique $T$-path between $\bottom(e_1)$ and $\mytop(e_2)$. Note that for $i \in \{1, 2\}$, we have $\Cov(e_i, f) = \CrossCov(e_i, f)$ since $e_i$ and $f$ are on different paths. We have that $\Cov(e_1, f) - \Cov(e_2, f) \le \CrossCov(e_1) - \CrossCov(e_2)$ since the LHS counts the number of cross-edges having one endpoint on $p$ and the other in the (maximal) subpath rooted at $\bottom(f)$, while the RHS counts the total number of cross-edges with an endpoint on $p$.

  Furthermore, we have that $\CrossCov(e_1) \ge \CrossCov(e_2)$ since every cross-edge covering $e_2$ must also cover $e_1$ (otherwise its endpoint would be on $p$ and it would not be a cross-edge).
  
  Finally, since $\Cov(e_1, f) > 1/5 \cdot \Cov(e_1)$ (weak interest), and $\Cov(e_2, f) \le 1/10 \cdot \CrossCov(e_2)$ (no $1/10$-interest), we have $\Cov(e_1, f) - \Cov(e_2, f) > 1/5 \cdot \CrossCov(e_1) - 1/10 \cdot \CrossCov(e_2) \ge 1/10 \cdot \CrossCov(e_2)$. Combining, we have:
  \begin{align*}
    \CrossCov(e_1) & = \CrossCov(e_2) + ( \CrossCov(e_1) - \CrossCov(e_2) ) \\
                   & \ge \CrossCov(e_2) + ( \Cov(e_1, f) - \Cov(e_2, f) ) \\
                   & \ge \CrossCov(e_2) + 1/10 \cdot \CrossCov(e_2) = (1 + 1/10) \cdot \CrossCov(e_2) .
  \end{align*}
  This proves the subclaim.

  \paragraph{Subclaim 2.} Each edge $e_{\text{fix}} \in E(P_i)$ is $1/10$-interested in at most $10$ paths.
  \paragraph{Proof of Subclaim 2.} We consider all cross-edges with one endpoint in $\subtree(e_{\text{fix}})$. If $e_{\text{fix}}$ is $1/10$-interested in $P_j$, then $1/10$-fraction of those edges must have their other endpoint in $P_j$. But since the other endpoint is  unique, there can be at most $10$ such different $P_j$s that $P_i$ is $1/10$-interested in. This proves the subclaim.

  \paragraph{Completing the proof using the subclaims.} Consider a path $P_i$ and let $e_1, \ldots, e_\ell$ be the edges of $P_i$ ordered from bottom-most (farthest away from the root) to top-most (closest to the root). We \emph{mark} all paths $P_j$ that $e_1$ is $1/10$-interested in. We iteratively consider $e_2, e_3, \ldots, e_\ell$ until we find an unmarked path $e_i$ is weakly interested in. At that point, we mark all $O(1)$ paths $e_i$ is $1/10$-interested in (Subclaim 2). Note that we can find such an edge $e_i$ with an unmarked weak interest at most $O(\log n)$ times since, due to Subclaim 1, each time we encounter such an edge $\CrossCov(e_i)$ increases by a multiplicative $1.1$-factor (and $\CrossCov(e_i)$ is at most the sum of weights over all edges, hence polynomially bounded). Therefore, since the markings can happen $O(\log n)$ times, and each time we mark $O(1)$ paths, at most $O(\log n)$ new paths can be marked. Finally, it is clear from construction that all weakly interested paths are marked, hence proving the claim.
\end{proof}

\subsection{Interest graph}\label{sec:interest-graph}
We now start making the structural path interest result algorithmic. We first compute the list of paths each path $P_i$ is interested in. The following structure formalizes the properties we require.
\begin{definition}\label{def:interest-list}
  Given a star instance $\{ P_i \}_{i=1}^k \subseteq T \subseteq G$, an \textbf{interest list} of a path $P_i$ is a list of path IDs $P_j$ such that (1) the list contains (the IDs of) all paths that $P_i$ is strongly interested in, (2) for each path $P_j$ in the list of $P_i$, we have that $P_i$ is (at least) weakly interested in $P_j$.
\end{definition}
In other words, an interest list of $P$ contains all paths that $P$ is strongly interested in, but may contain some additional paths that $P$ is only weakly interested in. Note that the size of any valid interest list is $\tl{O}(1)$ due to \Cref{lemma:path-weak-interest-cardinality}.

\begin{lemma}\label{lemma:interest-list-computation}
  Let $\{ P_i \}_{i=1}^k \subseteq T \subseteq G$ be a star instance. There is a deterministic $\tl{O}(1)$-round Minor-Aggregation algorithm after which (all nodes and edges on) $P_i$ learns its interest list.
\end{lemma}
\begin{proof}
  All nodes and edges can each $P_i$ can learn they are a part of $P_i$ using a single Minor-Aggregation round where all path edges are contracted and they can agree on an arbitrary ID of $P_i$. Note that each edge $e$ knows whether it is a cross-edge or not since its endpoints know in which path they are. Next, we assign special labels for all cross-edges. A cross edge $e := \{w_1 \in V(P_i), w_2 \in V(P_j)\}$ gets assigned (1) a label of (ID of) $j$ with a weight of $w(e)$, and (2) a label of (ID of) $i$ with a weight of $w(e)$. Nodes on $P_i$ will ignore label (2), i.e., only seeing ID $j$ and vice versa for $P_j$ which ignore label (1). Other edges do not get assigned a label.

  Next, each node $v$ of each path computes $O(1)$ (approximately) the most frequent labels (with respect to the weight as the multiplicity) among the edges in $\subtree(v)$. This is performed using the subtree sum task \Cref{lemma:Det-Ops} on each $P_i$ where the aggregation operator is the heavy-hitter operator (from \Cref{example:heavy-hitters}) with the parameter $h := 4$. %It is important to emphasize that both subtree sum and heavy hitter operations are aggregation operations, hence they can be efficiently performed $\tl{O}(1)$ rounds of Minor-Aggregation.
  By the guarantees of \Cref{example:heavy-hitters}, the list of heavy hitters for a path-node $v$ is guaranteed to contain all (IDs of) paths that $\{ \parent(v), v \}$ is strongly interested in, and each element of the list is an ID of a path that $\{ \parent(v), v \}$ is weakly interested in.

  Finally, we compute the union of the interest list of each path $P_i$: contract all path edges and simply use the union on the interest list as the aggregation operation. Since the size of (any union) of these lists is at most $\tl{O}(1)$ (\Cref{lemma:path-weak-interest-cardinality}), we can compute the union in $\tl{O}(1)$ rounds.
\end{proof}

After having access to the interest list, we construct a logical \emph{interest graph} between the paths where edges are created between mutually-interested pairs of paths.
\begin{definition}\label{def:hl-interest-graph}
  Let $\{ P_i \}_{i=1}^k \subseteq T \subseteq G$ be a star instance and suppose each $P_i$ knows its interest list. An \textbf{interest graph} $I$ is an undirected graph with $k$ nodes identified with the paths, i.e., $V(I) = \{P_1, \ldots, P_k\}$. There is an edge between $\{ P_i, P_j \} \in E(I)$ iff $P_i$ and $P_j$ are \textbf{mutually interested}, meaning that both $P_1$ is in the interest list of $P_2$ and vice versa.
\end{definition}

Next, we show one can efficiently simulate CONGEST algorithms on the (logical) interest graph. This mainly follows from the small maximum degree of interest graph.
\begin{lemma}\label{lemma:hl-interest-simulation}
  Let $\{ P_i \}_{i=1}^k \subseteq T \subseteq G$ be a star instance and suppose each $P_i$ knows its interest list. Any (deterministic) $\tau$-round CONGEST algorithm on the interest graph $I$ can be simulated in $O(\tau)$ (deterministic) rounds of Minor-Aggregation on $G$.
\end{lemma}
\begin{proof}
  In order to simulate a CONGEST algorithm on the interest graph $I$, it is sufficient to simulate a single round of $\tl{O}(1)$-bit communication between each pair of mutually-interested paths (each node and edge on the path learns all messages sent to that path via the algorithm on $I$). First, we contract all paths to a single node. Next, each path broadcasts its interest list and the messages intended for each entry on the interest list. Since there are $\tl{O}(1)$ entries on the list and each entry corresponds to $\tl{O}(1)$ bits, this can fit within a single Minor-Aggregation message. Next, each cross-edge $\{u \in P_i, v \in P_j\}$ can check whether its endpoints are mutually interested in each other (since it received their interest lists); if yes, the edge exchanges their messages with each other. Finally, each path (i.e., supernode) takes the union of all the messages received. Since each path can only receive messages from paths in its interest list, the size of this union is $\tl{O}(1)$, hence can be performed in a single round. This successfully simulates a CONGEST algorithm.
\end{proof}

\subsection{Solving the star instance via interest graph coloring}

In order to solve the star instance, the general idea will be to first edge-color the interest graph, process the color classes in series (giving us a matching between the paths), and then call the 2-respecting min-cut on each pair of matched paths. As a first step towards this goal, we need a classic edge-coloring result of Panconesi and Rizzi.
\begin{lemma}[\cite{panconesi2001some}]\label{lemma:edge-coloring-congest}
  Given a graph $G$ with maximum degree $\Delta = \tl{O}(1)$, there is a deterministic $\tl{O}(1)$-round CONGEST algorithm on $G$ that colors the edges into $\tl{O}(1)$ colors (where edges of each fixed color form a matching).
\end{lemma}

We have all the pieces in place to prove the main result of this section.
\thmStarAlgorithm*
\begin{proof}
  \textbf{Algorithm.} We first compute the 1-respecting cuts and remember the best result (\Cref{thm:1-respecting-cut-algorithm}). Next, we compute the interest list for each path (\Cref{lemma:interest-list-computation}) which defines the interest graph (\Cref{def:hl-interest-graph}). Next, we find an edge coloring for all cross-edges via the coloring CONGEST algorithm (\Cref{lemma:edge-coloring-congest}) by running it on the interest graph with the help of the simulation result (\Cref{lemma:hl-interest-simulation}). Since the maximum degree of the interest graph is at most $\tl{O}(1)$, this produces an coloring of cross-edges into $\chi \le \tl{O}(1)$ colors. Then, for each color class $c \in \{1, 2, \ldots, \chi\}$ we iteratively consider each pair of paths $\{P_i, P_j\}$ that are connected with a cross-edge of color $c$. By definition of a proper edge matching, all pairs of paths in the class $c$ are node-disjoint from other pairs in the same class. Furthermore, $G[ V(P_i) \cup V(P_j) ]$ is a connected graph (due to the existence of a cross-edge). Therefore, we can run simultaneous instances of Minor-Aggregation algorithms on each one of them. For each pair of such paths $\{P_i, P_j\}$, we construct a virtual node $r_{ij}$ and connect it to the top node of $P_i$ and $P_j$ with arbitrary weight (since it won't be considered); we run the 2-respecting path-to-path algorithm on each such pair (\Cref{thm:path-to-path-algorithm}). Due to the connectedness of $G[ V(P_i) \cup V(P_j) ]$, we can remove the virtual root node and recover an Minor-Aggregation on $G$ (\Cref{lemma:virtual-node-replacement}, with a $O(1)$-blowup in the number of rounds). The smallest cut ever seen is returned as the best 2-respecting cut. It is clear that every step takes at most $\tl{O}(1)$ rounds, hence the final algorithm takes $\tl{O}(1)$ rounds.

  \smallskip

  \noindent\textbf{Correctness analysis.} First, it is clear that the algorithm only finds feasible cuts (cannot return a cut smaller than possible). Hence, we only need to show it finds at least one optimal cut. The best 2-respecting cut is either a 1-respecting cut, or it occurs on a pair of paths $\{ P_{i^*}, P_{j^*} \}$ which are mutually interested in each other (\Cref{lemma:best-cut-has-interest}). However, there is an edge between $\{ P_{i^*}, P_{j^*} \}$ in the interest graph, hence we will find the best 2-respecting cut between $\{ P_{i^*}, P_{j^*} \}$ using the path-to-path algorithm (\Cref{thm:path-to-path-algorithm}), which will find the optimum. This completes the claim.
\end{proof}

\section{Between-Subtree 2-Respecting Min-Cut}\label{sec:2-respecting-between-subtree}

This section shows how to compute the minimum 2-respecting cut between $k$ subtrees $T_1, T_2, \ldots, T_k$ (adjoined with a root for orientation purposes, see \Cref{fig:between-subtree}). We call such an input a ``subtree instance'' and formalize it in the following definition.

\begin{definition}\label{def:subtree-instance}
  A \textbf{subtree instance} $\{ T_i \}_{i=1}^k \subseteq T \subseteq G$ is composed of the following. Suppose $G$ is a weighted graph and $T \subseteq G$ is $G$'s (rooted) spanning tree. Moreover, $T$ is composed of exactly a root $r$, and $k$ (disjoint) trees $T_1, T_2, \ldots, T_k$. 
\end{definition}

\begin{figure}
  \centering
  \includegraphics[width=0.18\textwidth]{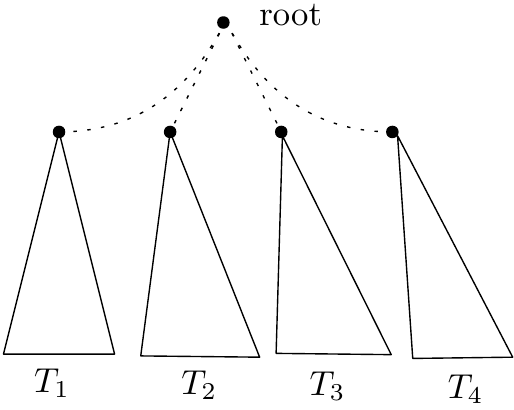}
  \caption{A subtree instance with $k = 4$ subtrees.}
  \label{fig:between-subtree}
\end{figure}

Our first idea is to reduce the problem for general $k$ to the case when $k = 2$. Suppose the optimum 2-respecting cut $(e^*, f^*)$ is contained in subtrees $e^* \in E(T_{i^*})$ and $f^* \in E(T_{j^*})$ where $i^* \neq j^*$. We want to find a way to break the symmetry between the subtrees $i^*$ and $j^*$, which we can do with the following structure.
\begin{definition}
  Given a universe of $k$ elements, a \textbf{pairwise coloring} is a collection $\{ f_1, \ldots, f_\chi \}$ where each $f_i : [k] \to \{ \mathrm{red}, \mathrm{blue} \}$ is called a \emph{color assignment} which assigns the color $f_i(j)$ to element $j$ and such that for all pairs $j \neq j' \in [k]$ there exists $i \in [\chi]$ such that $f_i(j) \neq f_i(j')$
\end{definition}
It is a folklore result that there exists such an assignment with $\chi = O(\log n)$. After constructing such a collection of colorings, we can iterate over each color assignment and for each assignment merge all the roots of all subtrees colored $\mathrm{red}$ and all subtrees colored $\mathrm{blue}$. This is, however, done implicitly in the proof of \Cref{thm:subtree-algorithm}.
\begin{lemma}\label{lemma:pairwise-coloring-construction} 
  Given a subtree instance $\{ T_i \}_{i=1}^k \subseteq T \subseteq G$, we can compute and distributedly store a pairwise coloring of the $k$ subtrees $\{T_i\}_{i=1}^k$ in $\tl{O}(1)$ deterministic Minor-Aggregation rounds.
\end{lemma}
\begin{proof}
  Each subtree can compute its (arbitrary, but unique) ID by contracting all edges in the subtree and computing the minimum ID of all nodes within it. This ID has $\chi := \tl{O}(1)$ bits. We iterate over each bit $i$ and create a new color assignment $f_i$ for each bit: $f_i(j) := \mathrm{blue}$ if the $i$'th bit of $j$'th ID is $0$ and $f_i(j) := \mathrm{red}$ otherwise.
\end{proof}

The following result formalizes the goal of solving the between-subtree cuts, and is the main result of this section. The proof is illustrated with \Cref{fig:subtree-reduction}.
\begin{restatable}{theorem}{thmSubtreeAlgorithm}\label{thm:subtree-algorithm}
  Given a subtree instance $\{ T_i \}_{i=1}^k \subseteq T \subseteq G$, there exists a deterministic $\tl{O}(1)$-round Minor-Aggregation algorithm on $G$ that computes the minimum of 1-respecting cuts and 2-respecting cuts $\min_{i < j} \allowbreak \min_{e \in E(T_i), \allowbreak f \in E(T_j)} \allowbreak \Cut_{T, G}(e, f)$.
\end{restatable}
\begin{proof}
  \textbf{Algorithm.} We first compute the 1-respecting cuts and remember the best result (\Cref{thm:1-respecting-cut-algorithm}). Furthermore, we construct a pairwise coloring $\{ f_i \}_{i=1}^\chi$ of $\{ T_i \}_{i=1}^k$ (\Cref{lemma:pairwise-coloring-construction}). Next, we iterate over all possibilities for (1) color assignment $i \in [\chi]$, (2) HL-depth $d_1 \le O(\log n)$, and (3) HL-depth $d_2 \le O(\log n)$. Next, all subtrees $T_j$ with $f_i(j) = \mathrm{red}$ will contract all edges $e \in E(T_j)$ where the $\HLdepth(e) \neq d_1$ and all subtrees $T_j$ with $f_i(j) = \mathrm{blue}$ will contract all edges $e \in E(T_j)$ where the $\HLdepth(e) \neq d_2$. This transforms the instance into a star instance (see \Cref{fig:subtree-reduction}). We use the star instance algorithm to find the best 2-respecting cut on this star (\Cref{thm:star-algorithm}). After iterating over all possibilities, the best result is returned. Note that, since there are $\tl{O}(1)$ color assignments, $O(\log n)$ HL-depths (\Cref{lemma:hl-has-logn-light-edges}), and the star algorithm takes $\tl{O}(1)$ rounds, the entire algorithm takes $\tl{O}(1)$ Minor-Aggregation rounds.

  \smallskip
  \noindent\textbf{Correctness analysis.} First, we note that the algorithm only checks some number of existing 2-respecting cuts, hence it can never report an answer that this is smaller than the optimum solution. We only need to show that it is successfully managed to find the optimum. Suppose that $(e^*, f^*) \in E(T_{a^*}) \times E(T_{b^*})$ with $a^* \neq b^*$ are the pair of edges that minimizes the 2-respecting cut. Let $d_1^* := \HLdepth(e^*)$ and $d_2^* := \HLdepth(f^*)$ and let $P_1^*$ and $P_2^*$ be the HL-paths which contain $e^*$ and $f^*$, respectively. Due to the definition of pairwise coloring, there exists a $i^* \in [\chi]$ where $f_{i^*}(a^*) \neq f_{i^*}(b^*)$. Therefore, when the algorithm iterates over $(i, d_1, d_2) = (i^*, d_1^*, d_2^*)$, it will transform the subtree instance $\{ T_i \}_{i=1}^k \subseteq T \subseteq G$ into a star instance $\{ P'_i \}_{i=1}^{k'} \subseteq T' \subseteq G'$ by contracting tree edges that are not on $P_1^*, P_2^*$. However, since contraction of edges not on $P_1^*$ or $P_2^*$ does not change the 2-respecting cut values, we have that $\Cut_{T, G}(e^*, f^*) = \Cut_{T', G'}(e^*, f^*)$. Therefore, since $\Cut_{T', G'}(e^*, f^*)$ will be considered in the star algorithm, the returned solution will return the optimum.
\end{proof}

\begin{figure}
  \centering
  \includegraphics[width=0.6\textwidth]{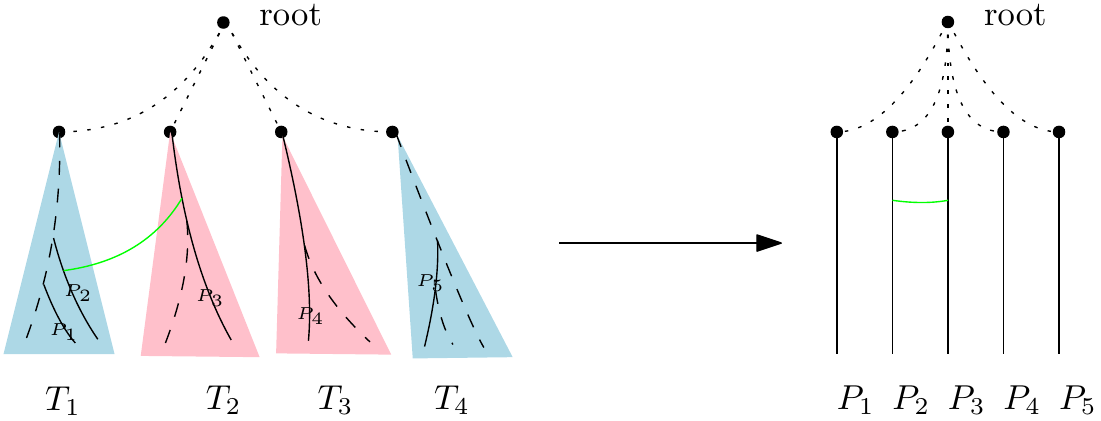}
  \caption{A depiction of the transformation from a subtree instance to a star instance for a particular pairwise coloring (two outer subtrees on the left are blue, and the two inner ones are red) and for a particular choice of $(d_1^*, d_2^*) = (0, 1)$. The green, between-subtree edge, is preserved in the star instance.}
  \label{fig:subtree-reduction}
\end{figure}

\section{General 2-Respecting Min-Cut}\label{sec:general-2-resp-cut}

In this section, we show how to compute the minimum 2-respecting on an arbitrary spanning tree $T$ of a weighted graph $G$.

%\thmGeneralRespectingCut*
\begin{restatable}{theorem}{thmGeneralRespectingCut}\label{thm:general-2-respecting-cut}
  Given a spanning tree $T$ of a weighted graph $G \supseteq T$, there exists a deterministic $\tl{O}(1)$-round Minor-Aggregation algorithm on $G$ that computes the minimum 2-respecting cut $\min_{e \in E(T), f \in E(T)} \allowbreak \Cut_{T, G}(e, f)$. This implies a deterministic $\tl{O}(D + \sqrt{n})$-round CONGEST algorithm for general graphs, and a deterministic $\tl{O}(D)$-round CONGEST algorithm for excluded-minor graphs.
\end{restatable}

Our general idea will be to choose a node $c \in V(T)$ and split the tree $T$ around this node into maximal connected subtrees $T_1, \ldots, T_k$. Then, it can happen that the optimal pair $e^*, f^*$ is either (1) in two different subtrees $T_{i^*}, T_{j^*}$, or (2) in the same subtree $T_{i^*}$. For case (1), we need to call the 2-respecting between-subtree cut algorithm on $T_1, \ldots, T_k$; for case (2) we will use recursion on each one of the (node disjoint) subtrees $T_i$, allowing us to schedule all recursive calls simultaneously. One important consideration of this approach is the depth of the recursion. By choosing the pivot point as the centroid, a node that splits the tree into balanced subtrees, this depth can be bounded by $O(\log n)$. The existence of the centroid is a well-known result, and it can be found in Minor-Aggregation using a subtree sum operation.
\begin{fact}[Folklore]\label{lemma:centroid-exists} 
  Any tree $T$ with $n$ nodes has a node $c$, called a \textbf{centroid}, whose removal leaves the remaining maximal connected components to have at most $|V(T)| / 2$ nodes.
\end{fact}

\begin{lemma}\label{lemma:find-centroid}
  Finding a centroid of a tree can be done in can be solved in $\tO(1)$ deterministic rounds of Minor-Aggregation.
\end{lemma}
\begin{proof}
  We root the tree $T$ arbitrarily (e.g., use a single round to contract all edges and find the node with the minimum ID). Each node $v$ finds the size of the subtree $s_v$ via a subtree sum operation (\Cref{lemma:Det-Ops}; each node sets its input to $1$ and uses the $+$-aggregator). In a single round, each node $v$ computes the sizes of the largest component in $T - v$ (when $v$ is deleted from $T$). This can be computed by finding the largest subtree size $S_c$ of direct children $c$ of $v$. The maximum between that value and $n - S_v$ is the size of the largest component of $T_v$. If this value has size at most $n/2$, $v$ declares itself the centroid. We do a single round of leader election between the centroids (e.g., by contracting all edges and taking the minimum ID between them).
\end{proof}

When performing the recursion, we ideally want private copies of $E(G)$ that are used to evaluate the 2-respecting cut values in each recursive call in order to prevent congestion issues caused by multiple recursive calls trying to use the same resources. In order to side-step this issue, we create cut-equivalent subgraphs that are disjoint with the help of virtual nodes. The procedure is depicted in \Cref{fig:centroid-decomposition}.
\begin{lemma}[Cut-equivalent subtrees]\label{lemma:cut-equivalent-subtrees}
  Given a spanning tree $T \subseteq G$ and a node $c \in V(T)$, let $T_1, \ldots, T_k$ be the maximal connected subtrees of $T - c$ and let $e_1, \ldots, e_k$ be the edges connecting $c$ with $T_i$ for each $i$. Let $T'_i := T[V(T_i) \cup \{c\}]$. We can construct and distributedly store graphs $H_1,  \ldots, H_k$ where (1) $T'_i$ is a spanning tree of $H_i$, (2) $H_i$ is an extension of $G[V(T_i)]$ with $O(1)$ virtual nodes, and (3) for all $i \in [k]$ and $e, f \in E(T'_i) = E(T_i) \cup \{e_i\}$ we have $\Cut_{T'_i, H_i}(e, f) = \Cut_{T, G}(e, f)$. The algorithm is deterministic and takes $\tl{O}(1)$ Minor-Aggregation rounds.
\end{lemma}
\begin{proof}
    %We argue there exists a partition , where $n := |V(T)|$. Start with $I_1 \gets I_2 \gets \emptyset$ and iterate over $i \in [k]$, in each step adding it to either $I_1$ or $I_2$ in a way that minimizes the absolute value of $n(I_1) - n(I_2)$. This guarantees that the absolute value of $n(I_1) - n(I_2)$ never exceeds $\max n_i \le n/2$ (). Furthermore, since $n(I_1) + n(I_2) \le n$, we have that $n \ge n(I_1) + n(I_2) = 2n(I_1) + [ n(I_2) - n(I_1) ] \ge 2 n(I_1) - n/2$, implying $n(I_1) \le \frac{3}{4} n$ (and analogously for $n(I_2)$). We note that the centroid ????????????? HOW TO COMPUTE THIS IN LOCAL!!! 

  We now describe how to construct $H_i$ and $T'_i$. We first set $H_i \gets G[V(T_i)]$ and, furthermore, add a virtual node $c_i$ that is connected to the non-$c$ to endpoint of $e_i$. Next, we consider all edges $e \in E(G)$ and do the following. If both of its endpoints are in $V(H_i)$, we continue (this edge is already included in $H_i$). We call such an edge a \emph{preserved edge}. If neither of the endpoints are in $V(H_i) \setminus \{ c \}$, we ignore this edge (it does not contribute to $H_i$). Finally, if exactly one endpoint (say $u$) of $e = \{u, v\}$ is in $V(H_i) \setminus \{c\}$, we add to $H_i$ an edge between $u$ and $c_i$ of weight $w(e)$; we call this edge a \emph{split edge}. Note that only $u$ needs to knows that it is connected to $c_i$, but this is consistent with distributedly storing a virtual graph as an extension of $G[V(T_i)]$. It is clear that the extension has a single virtual node.

  We now argue the cut equivalency, i.e., that $\Cut_{T'_i, H_i}(e, f) = \Cut_{T, G}(e, f)$. Let $e, f \in E(T'_i) = E(T_i) \cup \{e_i\}$. Then, $\Cut_{T, G}(e, f)$ is the sum of weights of all edges in $G$ which have exactly one endpoint on the $T$-supported shortest path between $e$ and $f$. Similarly, $\Cut_{T'_i, H_i}(e, f)$ is the sum of weights of all edges in $H_i$ which have exactly one endpoint on the $T_i$-supported path between $e$ and $f$. However, there is a natural one-to-one between such edges: each edge in $e \in E(H_i)$ is either a preserved or a split edge in $E(G)$; each preserved or split edge in $E(G)$ which has exactly one endpoint on the path also has exactly one endpoint after being preserved/split, thereby proving cut equivalency.
\end{proof}

\begin{figure}
  \centering
  \includegraphics[width=0.7\textwidth]{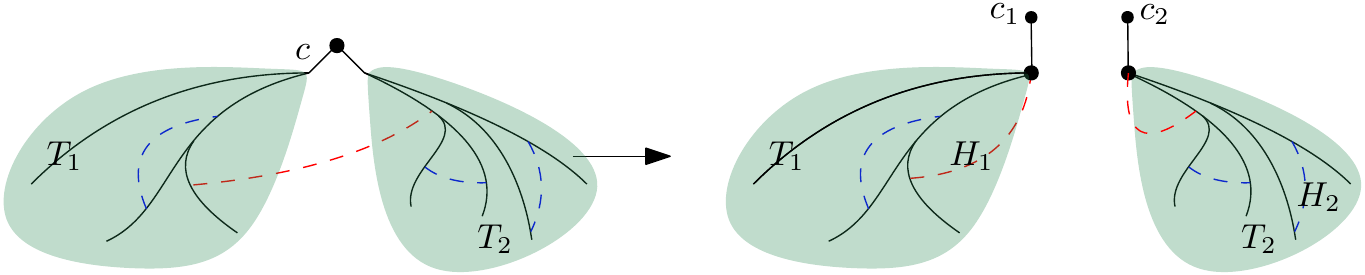}
  \caption{A depiction of disconnecting a centroid $c$ and creating private cut-equivalent subgraphs $H_1$ and $H_2$. The blue edges have both endpoints in the same subtree, hence they are preserved. The red edge is cross between subtrees, hence each subgraph creates its own copy of the red edge and reconnects its external endpoint to its virtual centroid $c_i$.}
  \label{fig:centroid-decomposition}
\end{figure}

Finally, with all the techniques in place, we are able to prove the main result of this section.
\thmGeneralRespectingCut*
\begin{proof}
  Let $G\orig$ be the initial communication network (we will construct various virtual network throughout the algorithm), let $T\orig$ be the initial spanning tree, and let $n := |V(G\orig)|$.

  We now describe a recursive procedure that will find the minimum 2-respecting cut on a weighted graph $G = (V, E_G)$ with respect to a tree $T = (V, E_T)$. Suppose, in some current recursive call we are considering, that that $G$ (and $T$) has $\beta$ many virtual nodes compared to $G\orig$ (i.e., as an extension of in the sense of \Cref{def:virtual-graph-extension}); initially, $\beta = 0$.

  Find a centroid $c$ of $T$, and suppose $T_1, \ldots, T_k$ are the maximal subtrees of $T - c$ and let $e_1, \ldots, e_k$ be the edges connecting $c$ with $T_i$ for each $i$. 

  Next, for each $i$ we create a virtual node $c_i$ representing the centroid and connect it to $T_i$ with the weight of $w(e_i)$; call this new tree $T'_i$. Furthermore,  we create a (common) virtual root $r$ and connect it to all $c_i$ (with an arbitrary weight); we call this tree $T''$. For technical purposes, we also define $G'$ where all edges adjacent to $c$ are subdivided into 2 sub-edges with the middle node corresponding to $c_i$ and the center node corresponding to $r$. We now call the between-subtree algorithm on the subtree instance $\{ T'_i \}_{i=1}^k \subseteq T'' \subseteq G'$ (\Cref{thm:subtree-algorithm}) and remember the best 2-respecting cut seen. While the algorithm operates on $T'' \subseteq G'$, we can eliminate the virtual node $r$ with a multiplicative $O(1)$ overhead (\Cref{thm:simulating-virtual-nodes}) since $G'$ is an extension of $G$ with a single virtual node.

  We construct and distributedly store $H_1, \ldots, H_k$ as in \Cref{lemma:cut-equivalent-subtrees}. Finally, we recursively (using the same procedure) find the minimum 2-respecting cut in all subtrees $T'_1 \subseteq H_1, T'_2 \subseteq H_2, \ldots, T'_k \subseteq H_k$. Due to the centroid guaranteeing $|V(T'_i)| \le |V(T)|/2 + 1$, the depth of this recursion is at most $O(\log n)$. With each recursive call we associate a virtual node. Specifically, for the call on $T'_i \subseteq H_i$ we associate the virtual node $c_i$. Furthermore, with a recursive call on $T \subseteq G$ we define the set $\mr{Virt} \subseteq V(G)$ as the union of all virtual nodes associated with itself and all (not necessarily direct) parent calls. Since each layer of recursion introduces a single node, we have that $|\mr{Virt}| = \beta \le O(\log n)$. Furthermore, by construction, we have that $G - \mr{Virt}$ and $T - \mr{Virt}$ are connected. Furthermore, by construction, $G - \mr{Virt} \subseteq G\orig$ and $T - \mr{Virt} \subseteq T\orig$.

  All branches of the recursion solve the 2-respecting cut in $T'_i \subseteq H_i$. Since $G - \mr{Virt}$ is connected, we eliminate the virtual nodes $\rm{Virt}$ with a $\tl{O}(1)$-blowup and recover an algorithm on $T - \mr{Virt} \subseteq G - \mr{Virt}$ (\Cref{thm:simulating-virtual-nodes}) that terminates in $C' \log_2^{C'} n$ rounds on $G$ (for some constant $C' > 0$). Note that, after elimination of the virtual nodes, all issued recursive calls on the same level are node disjoint, hence they can be scheduled together via \Cref{lemma:node-disjoint-scheduling}.

  \medskip

  \noindent\textbf{Runtime analysis.} The runtime of the entire algorithm being $\tl{O}(1)$ following immediately from the following facts: (1) the depth of the recursion is $O(\log n)$, (2) all work outside of the recursive calls is guaranteed to complete in $\tl{O}(1) \le C \log_2^C n$ rounds (where $C$ is some universal constant independent of the recursion analysis), (3) the recursive calls perform an algorithm on $G$ and all recursive calls on the same level of recursion are disjoint. This yields the $\tl{O}(1)$ runtime.

  \medskip
      
  \noindent\textbf{Correctness analysis.} Suppose that $e^*, f^*$ determine the 2-respecting minimum cut of $T \subseteq G$. Let $i_1^*, i_2^*$ be such that $e^* \in E(T'_{i_1^*})$, $f^* \in E(T'_{i_1^*})$ (they must exist since there is $\bigsqcup_i E(T'_i) = \bigsqcup_i E(T_i) \cup \{e_i\}$ is a partition of $E(T)$). We either have the case that $i_1^* = i_2^*$ in which case the recursive call on $T'_i \subseteq H_i$ combined with the cut equivalency $\Cut_{T'_i, H_i}(e^*, f^*) = \Cut_{T, G}(e^*, f^*)$ (\Cref{lemma:cut-equivalent-subtrees}) solves the problem by the recursive assumption. The other case, when $i_1^* \neq i_2^*$, this is solved in the subtree instance since the cuts there are trivially preserved.

  \medskip

  \noindent\textbf{Simulation in CONGEST.} We can directly simulate our deterministic Minor-Aggregation result in deterministic CONGEST using \Cref{thm:congest-simulation}.
\end{proof}

Finally, we can also prove the formal statements from the introductory section.
\begin{proof}[Proof of \Cref{theorem:final-min-cut-result}]
  We simply combine the $\poly(\log n)$-round tree packing Minor-Aggregation algorithm (\Cref{thm:treePacking}) with the $\poly(\log n)$-round 2-respective cut algorithm (\Cref{thm:general-2-respecting-cut}) to yield a $\poly(\log n)$-round Minor-Aggregation algorithm for the exact min-cut. Furthermore, we can compile down the Minor-Aggregation model to the CONGEST and obtain the guarantee for the algorithm (\Cref{thm:congest-simulation}). Finally, any correct algorithm for (even approximate) min-cut requires $\Omega(\SQ(G))$ rounds in CONGEST~\cite{haeupler2020shortcuts}, hence our algorithm is universally optimal (up to overhead factors).
\end{proof}

\paragraph{Acknowledgements.} We would like to thank Michal Dory, Bernhard Haeupler, and Richard Peng for helpful discussions.

% --- BIBLIOGRAPHY ---
\bibliographystyle{econ}
\bibliography{refs}

\appendix
\section{Deterministic Primitives}\label{sec:tree-primitives}

This section develops the deterministic Minor-Aggregation primitives for performing heavy-light decompositions, subtree sums and ancestor sums. The main technical idea used to achieve this result is by replacing the randomized star-merging technique used throughout the low-congestion shortcut framework with a deterministic version. Star-merging is a technique used in distributed computing, for example in Boruvka's algorithm, in which several adjacent node partitions are required to be merged together (requiring some data to be updated like the new leader or the ID of the new partition across all nodes). Randomized star-merging assigns to each node partition a label of either being a ``received'' or a ``joiner'' (using independent fair random coins)---joiners then merge into an adjacent receiver, guaranteeing that merging happens across star-like subgraphs (more formally, if we contract down each node partition to a single node and keep the edges across which the supernodes will merge, this graph will be a union of stars). Merging across star-like subgraphs is favorable since the diameter of the contracted graph is small. In order to derandomize star-merging, we leverage the deterministic 3-coloring of out-degree-one graphs developed by Cole and Vishkin~\cite{cole1986deterministic}. 

% \PaTaskNIO{task:numbered-path}{Numbered path prefix and suffix sum.}{A directed path $T$, stored as a rooted tree.}{Each node knows its ``path-depth'' (i.e., the hop-distance in $T$ from the root), has an $\tO(1)$-bit private input $x_v$, and the nodes agree on an aggregation operator $\bigoplus$.}{Each node $v$ learns the values $\bigoplus_{w \in A(v)} x_w$ and $\bigoplus_{w \in D(v)} x_w$, where $A(v)$ and $D(v)$ are the set of ancestors and descendants of $v$ in $T$.}

\begin{lemma}[Deterministic Star-Merging]
  \label{lemma:det-star-merging}
  Let $G = (V, \vec{E})$ be a $n$-node oriented graph without self-loops where the out-degree of each node is at most $1$, with $O \subseteq V$ having out-degree exactly $1$. Edges know their orientation. There is an $O(\log^* n)$-round Minor-Aggregation algorithm (communication is bidirectional) that partitions the nodes into $V = R \sqcup J$ (so-called receivers and joiners) such that (1) $|J| \ge 1/3 \cdot |O|$, (2) $J \subseteq O$, hence every $v \in J$ has a unique out-edge, and (3) for each $v \in J$, its out-edge points to a node in $R$. At termination, each node $v$ knows whether $v \in R$ or $v \in J$.
\end{lemma}
\begin{proof}
  Cole and Vishkin~\cite{cole1986deterministic} proposed a $O(\log^* n)$-round algorithm that 3-colors a graph with out degree at most one in the following communication model: in each round, each node $v$ broadcasts a $O(\log n)$-bit value $x_v$ that is received by nodes with their out-edge pointing to $v$. It is immediate to see that any $\tau$-round algorithm in this communication model can be simulated with a $\tau$-round algorithm in the Minor-Aggregation model.
    
  We run this 3-coloring algorithm in $O(\log^* n)$ Minor-Aggregation rounds on $G$, after which each node $v$ knows its colors $c_v \in \{0, 1, 2\}$. Using a single Minor-Aggregation round, we compute $N_k = |\{ v \in O : c_v = k \}|$, the number of nodes with out-degree $1$ of color $k$ for $k \in \{0, 1, 2\}$. This can be achieved by contracting all edges and using a sum-operator consensus step. Without loss of generality, let $N_0 \ge \max(N_1, N_2)$ be the color with most-frequent color. Therefore, $N_0 \ge |O|/3$.

  Nodes $v \in O$ (with out-degree $1$) with $c_v = 0$ are assigned to $J$. The rest, i.e., nodes with $c_v \neq 0$ or $v \not \in O$, are assigned to $R$. Since every two incident nodes have different color, each node in $J$ (color $0$) is pointing towards a node that is in $R$ (colors $1$ or $2$).
\end{proof}

With deterministic star-merging in place, we start building up higher-and-higher level primitives. We start off with a path prefix sum and suffix sum primitives, followed by ancestor and subtree sum under the assumption of a rooted tree, while finally developing a routing to orient a tree and removing the rootedness assumption.

\begin{lemma}[Numbered path prefix/suffix]\label{lemma:path-prefix-suffix}
  Let $G$ be a path graph on the nodes $(v_0, v_1, \ldots, v_{n-1})$. Suppose each node $v_i \in V(G)$ knows its index $i$ (hop-distance from $v_0$) and a $\tl{O}(1)$-bit private input $x_i$. Each node $v_k$ can compute the prefix and suffix aggregates $\bigoplus_{i=0}^{k} x_i$ and $\bigoplus_{i=k}^{n-1} x_i$, where $\bigoplus$ is some pre-defined $\tl{O}(1)$-bit aggregation operator. The algorithm is deterministic and terminates in $\tl{O}(1)$ Minor-Aggregation rounds on $G$.
\end{lemma}
\begin{proof}
  Every node can learn the number of nodes $n$ by contracting all edges and performing a consensus step with the sum-operator.
  
  We show how to compute the prefix sum $p_k = \bigoplus_{i=1}^{k-1} x_i$ of a node $v_k$; the suffix sum computation is analogous. We use recursion. We recursively solve the (same) problem on the first half (i.e., nodes with numbering $i$ such that $0 \le i < \lfloor n/2 \rfloor$) and recursively solve it on the remaining (second) half. Let $w$ be the largest-indexed node in the first half and let $p_w$ be its prefix sum (i.e., the aggregate of all inputs in the first half). In a single Minor-Aggregation round, we contract the second half of the path, as well as the interconnecting edge to broadcast $p_w$ to all nodes in the second half. Finally, each node $v_i$ in the second half performs $p_i \gets p_i \bigoplus p_w$ and outputs $p_i$ as its prefix sum.

  We analyze the algorithm. First, it is easy to verify that the above algorithm is correct: the prefix in the left half is correct by assumption, and the prefix in the second half is the one returned by the recursion aggregated with $p_w$. Second, we note that since every node knows its index and $n$, then the node knows how to translate the ``global algorithm'' from the previous paragraph to local actions. Third, since we apply the recursion on the first and second half (which are node-disjoint), we can run both branches of the recursion simultaneously (\Cref{lemma:node-disjoint-scheduling}). Since the depth of the recursion is $O(\log n)$, a total of $\tl{O}(1)$ Minor-Aggregation rounds suffice.
\end{proof}

\begin{lemma}[Heavy-light ancestor and subtree sums]\label{lemma:hl-subtree-sum}
  Let $T$ be a rooted tree (each node knows its parent). Furthermore, a heavy-light decomposition is known (each edge knows its HL-info) and each node $v$ has an $\tl{O}(1)$-bit private input $x_v$. There is a deterministic $\tl{O}(1)$-round Minor-Aggregation algorithm that computes for each node $v$ the values values $p_v := \bigoplus_{w \in A(v)} x_w$ and $s_v := \bigoplus_{w \in D(v)} x_w$, where $A(v)$ and $D(v)$ are the set of ancestors and descendants of $v$.
\end{lemma}
\begin{proof}
  We explain the subtree operation; the ancestor sum is completely analogous. We process the HL-paths in a bottom-up fashion by looping a variable $d$ from $O(\log n)$ to $0$. In iteration $d$, we compute the subtree sums for nodes of HL-depth equal to $d$ (note: the top-most node of an HL-path of HL-depth $d$ is processed in iteration $d-1$). Due to \Cref{lemma:hl-has-logn-light-edges}, each root-leaf path has at most $O(\log n)$ light edges, hence starting $d$ from $O(\log n)$ is sufficient.

  For a fixed $d$, the set of nodes with HL-depth equal to $d$ form a node-disjoint set of numbered paths, which are paths where each node knows its index in the path (i.e., hop-distance from the first node). Note that each node can deduce its HL-depth and its path-numbering from its HL-info. Therefore, the preconditions of \Cref{lemma:path-prefix-suffix} apply and we can compute the suffix sums on these paths. We use this in the following way. As private input $x_v$, each node $v$ sets $x_v := \bigoplus_{w \in \mathrm{NHC}(v)} s_w$, where $\mathrm{NHC}(v)$ is the set of non-heavy children of $v$. Computing $x_v$ can be implemented in a single Minor-Aggregation round by making each edge collect the HL-info and subtree sum $s$ from both of its endpoints, enabling each edge to deduce whether it is a heavy edge and passing the subtree sum to the parent node of non-heavy edges. We then compute the suffix sums on all nodes of HL-depth $d$ via \Cref{lemma:path-prefix-suffix}. It is easy to see that for each node $v$ of HL-depth $d$, the computed suffix sum at $v$ is the exact subtree sum at $v$, completing the computation in $\tl{O}(1)$ Minor-Aggregation rounds for a fixed $d$. Since $d$ takes on $O(\log n)$ different values, we conclude the entire computation requires $\tl{O}(1)$ Minor-Aggregation rounds.
\end{proof}

\begin{lemma}[Rooted heavy-light construction]\label{lemma:rooted-hl}
  Let $T$ be a rooted tree (each node knows its parent). There is a deterministic $\tl{O}(1)$-round Minor-Aggregation algorithm that constructs a heavy-light decomposition on $T$. Specifically, upon termination, each node learns its HL-info.
\end{lemma}
\begin{proof}
  We maintain a node partition $\calP = (P_1, \ldots, P_k)$ (specifically, $P_i \subseteq V(T)$, $P_i \cap P_j = \emptyset$ for $i \neq j$ and $\bigcup_i P_i = V(T)$) and in each part $P_i$ we maintain (1) that the induced subtree $T[P_i]$ is connected, (2) a valid heavy-light decomposition of the induced rooted subtree $T[P_i]$. Specifically, each edge $e \in E(T)$ maintains whether both of its endpoints are in the same part $P_i$ (called \emph{inside-part edges}) and each node maintains its HL-info with respect to $T[P_i]$. Initially, $\calP = \{ \{ v \} \mid v \in V(T) \}$ contains all singleton nodes and the task is trivial. We perform the following \emph{merging} iteration for $O(\log n)$ times.

  We denote with $T / \calP$ the minor of $T$ with all inside-part edges contracted (recall that Minor-Aggregation algorithms can freely operate on minors, \Cref{corollary:computation-on-minors}). First, each part $P_i$ marks its unique parent edge in $T / \calP$ ($T / \calP$ inherits the orientations from $T$) and suppose the marked edge is oriented from $P_i \in \calP$ to its parent $P_j \in \calP, j \neq i$. We call the part that contains the root of $T$ the \emph{root-part}, and the root-part does not mark any edges. With respect to the marked edges, each node in $T / \calP$ has out-degree at most $1$, and $|\calP| - 1$ of these nodes have out-degree exactly $1$ (all except the root part). Therefore, we can use \Cref{lemma:det-star-merging} to partition the set of parts $\calP = R \sqcup J$ (so-called \emph{receivers} and \emph{joiners}, resp.) such that $|J| \ge \frac{1}{3} ( |\calP| - 1 )$, each $P_i \in J$ has an out-edge to its parent $P_j, j \neq i$ that is a receiver (i.e., $P_j \in R$). Let $F \subseteq E(T)$ be the set of such edges, one for each $P_i \in J$ and pointing to a receiver, and the root-part is a receiver. Our goal is to merge the parts $\calP$ that form connected components along $F$ (specifically, maximal connected components in $T / \calP$, if we only consider edges in $F$). Note that such connected components are stars and joiners of a particular receivers are its descendants. Let $\calP' = \{ P'_1, P'_2, \ldots, P'_{k'} \}$ be the node partition representing the set of post-merge parts (each $P'_i$ corresponds to the union of parts $P_i$ in a single connected component along $F$).
  
  % For each $P_i$, the root of $T[P_i]$ randomly labels itself ``heads'' or ``tails'' by throwing a fair coin. This labeling is broadcast to every node $v \in P_i$ via a call to the part-wise aggregation oracle: exactly the edges that are internal to some $T[P_i]$ are labeled contracted, the private input to non-roots is set to $x_v = 0$, while the private input of the roots are set to a binary encoding of its heads/tails label (e.g., $x_v = 1$ if $v$ is a root that flipped heads and $x_v = 2$ in case of tails), the aggregation operator is simply the sum.

  % The root of each part $P_t$ that flipped tails considers the unique edge leading to its parent with respect to $T$. Let $P_h$ be the part of the parent. If $P_h$ flipped heads, then $P_t$ ``merges'' into $P_h$ while maintaining the stated invariant on $\calP$ after the merge (i.e., maintain a valid heavy-light decomposition). The merge is performed as follows.

  We recompute for each node $v \in P'_i$ its \emph{subtree size}, defined as $|V(\subtree(v))|$ with respect to its post-merge part $T[P'_i]$. Note that the previously-computed subtree sizes in the nodes of (pre-merge) receivers parts are correct, hence we only need to recompute it for $P'_i$ that are receivers. However, this can be achieved via a single heavy-light subtree sum call (\Cref{lemma:hl-subtree-sum}) on $P'_i \in R$: each node initializes its private input with one plus the sum of subtree sizes of its direct children that are in joiners (which are already computed). It is straightforward to verify that after the subtree sum call, each node $v \in P'_i$ has its subtree size correctly computed.

  Each node $v$ can now compute the child $u$ with the largest size, and label the edge $\{u, v\}$ heavy (all others are labeled light). In order to complete the merge, the invariant requires us to compute the remaining contents of HL-info: the depth and the list $L_v$ (containing IDs/depths of light edges on the root-to-$v$ path). Both can be achieved a single heavy-light ancestor sum computation (\Cref{lemma:hl-subtree-sum}): the depth is the result of an ancestor sum computation in $T[P'_i]$ with each node initializing their private input to $1$. The list $L_v$ is computed by an ancestor sum where each node $v$ whose parent is a light edge stores the information about this light edge as its private input and the rest leave the input empty; an ancestor sum is then performed that simply collects (concatenates) all of the private input. Since each node $v$ has only $O(\log n)$ light edges on the root-to-$v$ path, \Cref{lemma:hl-has-logn-light-edges} stipulates the final results fits within $\tl{O}(1)$ bits, rendering the operation valid.

  %The list $L_v$ is first computed for $P_h$ and then for the tails-parts in parallel (since they are disjoint). For $P_h$, we do a heavy-light ancestor sum (\Cref{task:hl-prefix}): the aggregation operator simply concatenates its inputs together, the private input of $v$ is the description of its parent edge if its light (or an empty list otherwise). Note that HL-info can be described with $\tl{O}(1)$ bits (\Cref{lemma:hl-info-small-size}), satisfying the $\tl{O}(1)$-bit requirement for the aggregation operator.

  %After the task completes, we do the analogous operation on the tail-parts $P_t$ with the root of $P_t$ inheriting the list from its parent in $P_h$. Note that each iteration used $O(1)$ part-wise aggregations and $O(1)$ heavy-light ancestor sums (the tasks on different disjoint parts are aggregated together), both of which can be implemented in $\tl{O}(1)$ CONGEST\PA rounds. Therefore, the algorithm can be implemented in $\tl{O}(1)$ CONGEST\PA rounds.

  It is easy to verify that the invariants are maintained via this procedure since only joiner-part merge into receiver-parts which are their the joiners' parents. Furthermore, the number of parts $|\calP|$ decreases in each iteration by a constant factor, hence after $O(\log n)$ iterations we have that $\calP = \{ V(T) \}$. In conclusion, all steps take $\tl{O}(1)$ Minor-Aggregation rounds.%
  %Fix an iteration and let $\calP$ and $\calP'$ be the partitions at the start and end of the iteration. Every part that does not contain the global root of $T$ will be deleted with probability at least $1/4$: it will be deleted if it flips tails and its parent part (part corresponding to the parent of the root of the part) flips heads. Therefore, $\E[|\calP'|] \le 1 + \frac{3}{4}(|\calP| - 1)$. Therefore, by standard arguments, we deduce $|\calP| = 1$ that after $O(\log n)$ rounds, w.h.p.
\end{proof}

\begin{theorem}[Orienting a tree]\label{lemma:orienting-and-hl-construction}
  Let $T$ be a (unrooted) tree and let $r \in V(T)$ be an arbitrary node (each node knows whether it is the root). There is a deterministic $\tl{O}(1)$-round Minor-Aggregation algorithm that constructs a heavy-light decomposition of $T$ rooted at $r$. Specifically, upon termination, each node learns its HL-info.
\end{theorem}
\begin{proof}
  We maintain a node partition $\calP = (P_1, \ldots, P_k)$ (specifically, $P_i \subseteq V(T)$, $P_i \cap P_j = \emptyset$ for $i \neq j$ and $\bigcup_i P_i = V(T)$) and in each part $P_i$ we maintain (1) that induced subtree $T[P_i]$ is connected, (2) a root $r_i$ of $T[P_i]$ which is the global root $r$ if $r \in P_i$ and is otherwise arbitrary, (3) consistent edge orientations on tree edges of $T[P_i]$ with respect to $r_i$, and (3) a valid heavy-light decomposition of $T[P_i]$ with respect to $r_i$. Specifically, each edge $e \in E(T)$ maintains whether both of its endpoints are in the same part $P_i$ (called \emph{inside-part edges}) and each node maintains its HL-info with respect to $T[P_i]$. Initially, $\calP = \{ \{ v \} \mid v \in V(T) \}$ contains all singleton nodes and the task is trivial. We perform the following \emph{merging} iteration for $O(\log n)$ times.

  We denote with $T / \calP$ the minor of $T$ with all inside-part edges contracted (recall that Minor-Aggregation algorithms can freely operate on minors, \Cref{corollary:computation-on-minors}). First, each part $P_i$ chooses and marks an arbitrary adjacent edge $e_i \in E(T)$ that is not an inside-part edge (i.e., connects $P_i$ to another part $P_j, j \neq i$). Marking such an edge in $T / \calP$ can be performed in a single Minor-Aggregation round. We oriented the edge $P_i$ marked as going from $P_i$ to $P_j$. We call the part that contains the root of $T$ the \emph{root-part}, and the root-part does not mark any edges. With respect to the marked edges, each node in $T / \calP$ has out-degree at most $1$, and $|\calP| - 1$ of these nodes have out-degree exactly $1$ (all except the root part). Therefore, we can use \Cref{lemma:det-star-merging} to partition the set of parts $\calP = R \sqcup J$ (so-called \emph{receivers} and \emph{joiners}, resp.) such that $|J| \ge \frac{1}{3} ( |\calP| - 1 )$, each $P_i \in J$ has an out-edge to some part $P_j, j \neq i$ that is a receiver (i.e., $P_j \in R$) Let $F \subseteq E(T)$ be the set of such edges, one for each $P_i \in J$ and pointing to a receiver, and the root-part is a receiver. Our goal is to merge the parts $\calP$ that form connected components along $F$ (specifically, maximal connected components in $T / \calP$, if we only consider edges in $F$). Note that such connected components are stars. Let $\calP' = \{ P'_1, P'_2, \ldots, P'_{k'} \}$ be the node partition representing the set of post-merge parts (each $P'_i$ corresponds to the union of parts $P_i$ in a single connected component along $F$).

  % For every $P_i$, the root randomly labels itself ``heads'' or ``tails'' by throwing a fair coin. This labeling is broadcasted to every node $v \in P_i$ via a call to the part-wise aggregation oracle: exactly the edges that are internal to some $T[P_i]$ are labeled contracted, the private input to non-roots is set to $x_v = 0$, while the private input of the roots are set to a binary encoding of its heads/tails label (e.g., $x_v = 1$ if $v$ is a root that flipped heads and $x_v = 2$ in case of tails), the aggregation operator is simply the sum.

  % The parts $P_t$ that flipped tails choose an arbitrary edge $e_t$ crossing into a part $P_h$ that flipped heads (or $e_t = \bot$ if such edge does not exist). This can be done via a single call to the part-wise aggregation oracle: edges internal to some $T[P_t]$ are contracted, the private input to a node $v \in P_t$ whose part flipped tails is (an encoding of) an arbitrary incident edge crossing into a part that flipped heads (or $\bot$ if none), the aggregation operator selects the lexicographically minimum input.

  The joiner part $P_i \in J$ considers its outgoing edge denoted as $e_i$ and merge with the part $P_j \in R$ on the other side of $e_i$. This requires us to maintain the invariants on $\calP'$. First, we floor the root $r_j$ of $P_j$ to all nodes in (all corresponding receivers) $P_i$. Note that the global root $r$ always remains the root of the part containing this since the root-part is always a receiver.

  One of the potentially violated invariants requires us to maintain the edge orientations on $T[P_i \cup P_j]$, that might not correspond to a properly rooted tree. We rectify this by keeping the orientations of $P_j$ unchanged and correcting the orientations of $P_i$ (note that multiple tails-parts can merge into the same heads-part). Specifically, suppose that $e_i = \{u_i, u_j\} \in E(T)$, $u_i \in P_i$, $u_j \in P_j$. We need to reverse the orientation of the edges on the path in $P_i$ between the root of $P_i$ and $u_i$. This can be achieved by constructing an heavy-light decomposition on $P_i$ via \Cref{lemma:rooted-hl} (since $P_i$ is rooted) and identifying each edge to be reverse in the following way. The parent edge of a node $v$ should be reversed if its depth is at most the depth of $u_i$ and the LCA of $v$ and $u_i$ is $v$ (which can be obtained from HL-info and \Cref{lemma:lca-labels}). It is easy to verify that this maintains the orientation invariant.

  It is easy to verify that the invariants are maintained via this procedure since only joiner-part merge into receiver-parts which are their the joiners' parents. Furthermore, the number of parts $|\calP|$ decreases in each iteration by a constant factor, hence after $O(\log n)$ iterations we have that $\calP = \{ V(T) \}$. In conclusion, all steps take $\tl{O}(1)$ Minor-Aggregation rounds.
\end{proof}

\lemmaDetOps*
\begin{proof}  
  We orient and construct a heavy-light decomposition via \Cref{lemma:orienting-and-hl-construction}. Then, we use \Cref{lemma:hl-subtree-sum} to compute the prefix and subtree sums.
\end{proof}

\thmCongestSimulation*
\begin{proof}
  % In order to avoid introducing an excessive amount of notation, we argue this using the terminology of \cite{goranci2022universally}.
  
  The randomized claims were explicitly argued in \cite{goranci2022universally}. We now argue the deterministic claims. First, we define the part-wise aggregation (PA) problem: we are given a set of disjoint and connected parts $P_1, P_2, \ldots, P_k$ where $P_i \subseteq V(G)$ (each node knows its part-ID) and each node is given a private $O(\log n)$-bit input $x_v$. The goal is for each node in part $P_i$ to learn the value $\bigoplus_{w \in P_i} x_w$. Prior work has shown that we can solve PA in deterministic $\tl{O}(D + \sqrt{n})$ rounds in general graphs~\cite{haeupler2016low} and in $\tl{O}(D)$ rounds in excluded-minor graphs.~\cite{haeupler2016low,ghaffari2021excluded}. Therefore, it is sufficient to show that we can simulate a single Minor-Aggregation round in $\tl{O}(Q)$ rounds of deterministic CONGEST, where $Q$ is the number of rounds it takes to solve the part-wise aggregation problem.

  To prove this, we can directly follow the proof from \cite{goranci2022universally} of Theorem 4.2. The only randomized step of the proof is the leader election step which uses a randomized star-merging technique. However, we can immediately substitute this step with the deterministic star-merging of \Cref{lemma:det-star-merging}, leading to a deterministic simulation result.

\end{proof}

\end{document}